\documentclass[11pt]{article}
\usepackage{fullpage}
\usepackage{parskip}

\usepackage{amsthm,amsmath,amssymb,booktabs,xcolor,graphicx}
\usepackage{thmtools,thm-restate}
\usepackage{bbm}
\usepackage{listings}
\usepackage{xcolor}
\usepackage{appendix}
\usepackage{enumerate}
\usepackage{mathbbol}
\usepackage[numbers]{natbib} 
\usepackage[shortlabels]{enumitem}

\usepackage[margin=1in]{geometry}
\usepackage[utf8]{inputenc}
\usepackage[T1]{fontenc}

\usepackage[textsize=small,backgroundcolor=orange!20]{todonotes}
\usepackage[hidelinks]{hyperref}
\usepackage{url}
\usepackage{etoolbox}
\usepackage{appendix}
\usepackage{xspace}
\usepackage[linesnumbered,ruled,vlined]{algorithm2e}
\usepackage[noend]{algpseudocode}
\usepackage[labelfont=bf]{caption}
\usepackage[noabbrev,capitalize]{cleveref}
\crefname{equation}{}{} 
\AtBeginEnvironment{appendices}{\crefalias{section}{appendix}} 

\usepackage[color,final]{showkeys} 

\colorlet{refkey}{orange!20}
\colorlet{labelkey}{blue!30}

\numberwithin{equation}{section}
\newtheorem{theorem}{Theorem}[section]

\newtheorem{lemma}[theorem]{Lemma}

\crefname{claim}{Claim}{Claims}

\newtheorem{corollary}[theorem]{Corollary}

\newtheorem*{question*}{Question}

\theoremstyle{definition}
\newtheorem{definition}[theorem]{Definition}

\newtheorem*{definition*}{Definition}
\newtheorem{example}[theorem]{Example}

\theoremstyle{remark}

\newcommand{\assign}{\leftarrow}

\newcommand{\norm}[1]{\left\lVert#1\right\rVert}

\newcommand{\one}{\mathbbm{1}}

\newcommand{\poly}{\mathrm{poly}}

\newcommand{\potential}{\phi}

\newcommand{\eps}{\varepsilon}

\newcommand{\E}{\mathbb{E}}

\newcommand{\R}{\mathbb{R}}
\newcommand{\calP}{\mathcal{P}}
\newcommand{\calA}{\mathcal{A}}
\newcommand{\calB}{\mathcal{B}}

\newcommand{\polylog}{\mathrm{polylog\ }}

\newcommand{\1}[1]{\mathbb{1}[#1]}

\title{Universally-Optimal Distributed Shortest Paths and Transshipment via Graph-Based $\ell_1$-Oblivious Routing\footnote{The author ordering was randomized using \url{https://www.aeaweb.org/journals/policies/random-author-order/generator}. The authors of this paper encourage citations by listing the authors with \texttt{\textbackslash textcircled\{r\}} instead of commas: Zuzic\textcircled{r}Goranci\textcircled{r}Ye\textcircled{r}Haeupler\textcircled{r}Sun.}}

\author{
  Goran Zuzic\thanks{Supported in part by the Swiss National Foundation (project grant 200021-184735).}\\
  \small ETH Zurich\\
  \small goran.zuzic@inf.ethz.ch

  \and\textcircled{r}\and
  
  Gramoz Goranci\\  
  \small School of Computing Science,\\
  \small University of Glasgow \\
  \small gramoz.goranci@glasgow.ac.uk

  \and\textcircled{r}\and
  
  Mingquan Ye\\
  \small University of Illinois at Chicago\\
  \small mye9@uic.edu

  \and\textcircled{r}\and  
    
  Bernhard Haeupler\thanks{Supported in part by NSF grants CCF-1527110, CCF-1618280, CCF-1814603, CCF-1910588, NSF CAREER award CCF-1750808, a Sloan Research Fellowship, funding from the European Research Council (ERC) under the European Union's Horizon 2020 research and innovation program (ERC grant agreement 949272), and the Swiss National Foundation (project grant 200021-184735).} \\
  \small ETH Zurich \& Carnegie Mellon University\\
  \small haeupler@cs.cmu.edu

  \and\textcircled{r}\and  
  
  Xiaorui Sun\thanks{Supported by start-up funds from University of Illinois at Chicago.}\\
  \small University of Illinois at Chicago\\
  \small xiaorui@uic.edu
}

\newcommand{\defeq}{\stackrel{\mathrm{\scriptscriptstyle def}}{=}}
\newcommand{\opt}{\mathrm{OPT}}

\newcommand{\dist}{\operatorname{dist}}
\newcommand{\imax}{{i_{\max}}}

\newcommand{\parent}{\mathrm{p}}

\newcommand{\shortcutquality}{\mathrm{Shortcut}\-\mathrm{Quality}}
\newcommand{\diameter}{\mathrm{HopDiameter}}
\newcommand\numberthis{\addtocounter{equation}{1}\tag{\theequation}}

\date{}
\begin{document}

\maketitle
\thispagestyle{empty}
\begin{abstract}
We provide universally-optimal distributed graph algorithms for $(1+\eps)$-approximate shortest path problems including shortest-path-tree and transshipment.

\medskip

The universal optimality of our algorithms guarantees that, on any $n$-node network $G$, our algorithm completes in $T \cdot n^{o(1)}$ rounds whenever a $T$-round algorithm exists for $G$. This includes $D \cdot n^{o(1)}$-round algorithms for any planar or excluded-minor network. Our algorithms never require more than $(\sqrt{n} + D) \cdot n^{o(1)}$ rounds, resulting in the first sub-linear-round distributed algorithm for transshipment.

\medskip

The key technical contribution leading to these results is the first efficient $n^{o(1)}$-competitive linear $\ell_1$-oblivious routing operator that does not require the use of $\ell_1$-embeddings. Our construction is simple, solely based on low-diameter decompositions, and---in contrast to all known constructions---directly produces an oblivious flow instead of just an approximation of the optimal flow cost. This also has the benefit of simplifying the interaction with Sherman's multiplicative weight framework [SODA'17] in the distributed setting and its subsequent rounding procedures.  
\end{abstract}

\newpage

\tableofcontents
\thispagestyle{empty}

\newpage

\section{Introduction}\setcounter{page}{1}

Computing single-source shortest paths (SSSP) is one of the most fundamental and well-studied problems in combinatorial optimization. Likewise, the distributed version of the SSSP problem has received wide-reaching attention in the distributed community~\cite{Nan14,EN16,HKN16,BKKL17,Elk17,FN18,GL18,HL18,EN19,LPP19,CM20}. For $(1 + \eps)$-approximate SSSP, this effort culminated in a distributed algorithm that is guaranteed to complete in $\tilde{\Theta}(\sqrt{n} + D)$ rounds\footnote{We use $\sim$, as in $\tilde{O}$, to hide $\poly(\log n)$ factors, where $n$ is the number of nodes in the network.}~\cite{BKKL17} on every $n$-node  network with hop-diameter $D$ in the standard message-passing model (CONGEST).

This $\tilde{\Theta}(\sqrt{n} + D)$-round algorithm is \emph{existentially-optimal}, in the sense that there exists a family of pathological networks where one cannot do better~\cite{sarma2012distributed}. However, the barrier that precludes fast algorithms in these pathological networks does not apply to most networks of interest, and faster algorithms often exist in non-worst-case networks such as planar graphs~\cite{GH16}. This motivates a much stronger notion of optimality---called \emph{universal optimality}---which requires a single (i.e., uniform) algorithm to be as fast as possible on \emph{every} network $G$, i.e.,  (approximately) as fast as the running time of any other algorithm on $G$. In particular, for any (unknown) network $G$ the universally-optimal algorithm must be competitive with whichever algorithm happens to be the fastest for $G$, including any algorithm explicitly designed to be fast on $G$ (but potentially very slow on any other network). More formally, let $T_{\calA}(G)$ be the running time on the network $G$ of the longest-running (i.e., worst-case) problem-specific input to an algorithm $\calA$. An algorithm $\calA$ is ($n^{o(1)}$-) universally optimal if $T_\calA(G)$ is $n^{o(1)}$-competitive with $T_\calB(G)$ for every other correct algorithm $\calB$, i.e., if $T_\calA(G) \le n^{o(1)}\cdot T_\calB(G)$~\cite{HWZ21}.

This beyond-worst-network guarantee is therefore in some sense the strongest form in which an algorithm can adjust (on the fly) to the network topology it is run on. The concept of universal optimality was already (informally) proposed in 1998 by Garay, Kutten and Peleg~\cite{GKP98} but only over the last six years work towards universally-optimal distributed algorithms has found traction, in the form of the low-congestion shortcut framework~\cite{GH16,haeupler2016low,haeupler2016near,haeupler2018minor,KKOI19,ghaffari2020low,HWZ21}. This framework can be used to design universally-optimal distributed algorithms for problems that can be solved  very fast using a communication primitive called part-wise aggregation. For example, such universally-optimal algorithms exist for (exact) minimum spanning tree (MST) and $(1 + \eps)$-minimum cut. The shortcut framework also implies fast algorithm with concrete guarantees for special graph classes, including any excluded-minor graph family. We refer to \Cref{sec:prelims} and \cite{HWZ21} for more details on universal optimality and the shortcut framework but remark that SSSP algorithms have been notoriously hard to achieve within the shortcut framework. The only non-trivial SSSP algorithm~\cite{HL18} within the framework has a bad super-constant approximation factor of at least $\polylog n$ (or even $n^{o(1)}$). An improvement using \cite{BKKL17} towards better approximations was proposed and stated as the main open problem in the abstract and conclusion of \cite{HL18}, however this suggested approach turned out to be (provably) impossible (see below).

\textbf{Our results.} In this paper, we resolve the open question of \cite{HL18} by giving a universally-optimal distributed algorithms for $(1+\eps)$-approximate single-source shortest path (SSSP) whenever there exists a $n^{o(1)}$-round algorithm. The running time of our algorithm is $n^{o(1)}$-competitive with the fastest possible correct algorithm and thus leads to ultra-fast sub-polynomial-round distributed algorithms on networks of interest. For example, our SSSP algorithm provably completes in $D \cdot n^{o(1)}$ rounds on any minor-free network.

We also give the first universally-optimal distributed algorithms for $(1+\eps)$-approximate transshipment.
Transshipment is a well-studied generalization of the shortest path problem also known as uncapacitated min-cost flow, earth mover's distance, or Wasserstein metric. No sub-linear-round algorithm was known for $O(1)$- or $(1+\eps)$-approximate transshipment in the distributed setting (CONGEST). Therefore even our worst-case running time of $(\sqrt{n} + D) \cdot n^{o(1)}$ rounds for $(1+\eps)$-approximate distributed transshipment improves over the state-of-the-art $\tilde{O}(n)$-round CONGEST algorithm~\cite{BKKL17} for this problem.


\textbf{Challenges.} All approaches to the distributed $(1+\eps)$-approximate SSSP problem that appear in the literature can be categorized as follows: (1) Either they select $\tilde{\Theta}(\sqrt{n})$ nodes (e.g., via random sampling), construct a backbone graph on the selected nodes in $\tilde{\Theta}(\sqrt{n})$ rounds, and finally solve some variant of the shortest path problem on the selected nodes. (2) Or they simulate all-to-all communications on a large number of nodes. 

Algorithms using hop-sets and similar ideas like \cite{Nan14,EN16,HKN16,Elk17,EN19,LPP19,CM20} are of the first type. The algorithm of \cite{BKKL17} actually fall into both categories simultaneously: \cite{BKKL17} explicitly uses the $\tilde{\Theta}(\sqrt{n})$ reduction to a backbone graph from \cite{Nan14} and then simulates a fast BROADCAST CONGESTED CLIQUE algorithm on the $\tilde{\Theta}(\sqrt{n})$ backbone nodes, which can be easily achieved in CONGEST with a slow-down of $\tilde{\Theta}(\sqrt{n})$ (via flooding). The algorithms from \cite{GL18b} are the cleanest example of the last category. They simulate almost any near-linear time PRAM algorithm (for SSSP, transshipment, or almost any other problem)  by implementing all-to-all communications between a linear number of nodes (in this case serving any communication which does not involve any node too often). This is possible in graphs with good expansion with a slowdown linear in the mixing time of the graph but impossible to do in sub-linear time in general graphs.

Once one looks at these algorithms in this way it is clear that these approaches fundamentally cannot \emph{ever} lead to a  $o(\sqrt{n})$ runtime, even when run on a simple network where one would/should expect $\tilde{o}(\sqrt{n})$-round algorithms to exist.

Our paper uses a different approach to SSSP, which requires us to solve the more-general transshipment problem. On a high level, transshipment asks us to find the minimum cost flow which satisfies a given demand $d \in \R^V$ without any capacity constraints on edges. For example, the $s$-$t$ shortest path corresponds to transshipment with the demand $d(v) = (\1{v=s} - \1{v=t})$. We solve the transshipment problem using \emph{Sherman's framework}~\cite{She17b} which yielded great success in the sequential and parallel settings~\cite{Li20,AndoniSZ20}. In particular, we show that Sherman's framework for $\ell_1$ can be implemented in CONGEST within the shortcut framework. The main technical barrier preventing Sherman's framework from being utilized in the distributed setting was the construction of a crucial object called \emph{(linear) $\ell_1$-oblivious routing (cost approximator)}. Concretely, the $\ell_1$-oblivious routing is a linear operator $R$ (i.e., a matrix) that takes the demand $d \in \R^V$ as input, and outputs a vector $R d$ such that $\norm{R d}_1$ is a good approximation to the (cost of the) optimal solution. All known approaches of constructing this matrix relied on first embedding the graph in $\ell_1$ space and then routing the demand over this space. Unfortunately, while such approaches are well-suited for shared-memory settings, the approach is fundamentally broken in the distributed world where arbitrary data shuffling cannot be done efficiently. In more detail, all known oblivious routing algorithms route the demand of a vertex $v$ to an intermediary vertex that depends on the $\ell_1$-embedding coordinate of $v$; this seems infeasible to do at scale in a message-passing context.

\textbf{Technical contribution.} Our main technical contribution is a new construction of $n^{o(1)}$-approximate $\ell_1$-oblivious routings that is graph-based and can be efficiently implemented in the distributed setting. The construction is simple and relies on routing the demand along several low-diameter decompositions (LDDs) in a bottom-up way. 

Our construction is particularly appealing in that it produces a \emph{flow vector} that routes (i.e., satisfies) the given demand and whose cost is $n^{o(1)}$-competitive with the optimal flow that satisfies the demand. This is in contrast to previous ``oblivious routings'' which routed the demand over the fictional $\ell_1$ space in a way that cannot be easily pulled back to a near-optimal flow in the underlying graph. In other words, previous constructions would be more aptly named \emph{cost approximators}, rather than oblivious routings, since they do not produce a flow in the graph. This can be compared to $(1+\eps)$-approximate maximum flow computation, where many fast algorithms construct a \emph{congestion approximator} which approximates the solution~\cite{She13,Pen16}; it is known that fast construction of an oblivious routing with the same property (i.e., a flow with near-optimal congestion satisfying the demand) is a significantly harder task.

To simplify the presentation of distributed algorithms in the shortcut framework, we introduce a new interface to this framework called the \emph{Distributed Minor-Aggregation model}. On one hand, the interface restricts the operations the nodes can perform in that they can only compute aggregates of adjacent nodes. On the other hand, the Distributed Minor-Aggregation model provides powerful high-level primitives like graph contractions which allow for succinct descriptions of otherwise complicated distributed algorithms. Our Minor-Aggregation model interface is our best attempt at making the most recent advancements in theoretical distributed computing and universal optimality~\cite{HWZ21,ghaffari2021hop} easily accessible. The idea is that any fast algorithm in the Minor-Aggregation model for a vast class of problems can be converted in a black-box way to a universally-optimal distributed algorithm in the standard CONGEST model. This gives rise to the following  informal, but helpful, claim about (non-local) distributed algorithms: \emph{An algorithm can be efficiently distributed if it can be written as only computing aggregates over a minor of the original network}. Our model connects nicely with similar intuitions or (implicit) models that have appeared in the distributed literature (e.g., \cite{GKKLP15,forster2020minor}). The following informal theorem shows how to convert running times from the Minor-Aggregation model to the standard CONGEST model.

\begin{theorem}[Informal, Minor-Aggregation model simulation]\label{theorem:informal-dmam-simulation}
  Any Minor-Aggregation algorithm on $G$ that terminates in $\tau$ rounds for SSSP, transshipment, MST, etc., can be converted into a $\tau \cdot n^{o(1)}$-competitive universally-optimal algorithm in the CONGEST model on $G$. The algorithm is guaranteed to complete in $D \cdot \tau \cdot n^{o(1)}$ CONGEST rounds when $G$ is a planar, genus-bounded, treewidth-bounded, excluded-minor graph, or an $n^{-o(1)}$-expander.
\end{theorem}

\textbf{Formal statements of our results.} With the required preliminaries in place, we can now formally state our results. The following theorems concern $(1+\eps)$-approximate distributed SSSP and transshipment in the Minor-Aggregation model. 
\begin{restatable}{theorem}{thmSSSP}\label{thm:SSSP}\textnormal{(Distributed SSSP).}
  Given an undirected and weighted graph $G$ with edge weights in $[1, n^{O(1)}]$, there exists a distributed algorithm that computes $(1+\eps)$-approximate SSSP in $\eps^{-2} \cdot \exp(O(\log n \cdot \log \log n)^{3/4})$ Minor-Aggregation rounds.
\end{restatable}
\begin{restatable}{theorem}{theoremDistributedTransshipment}
  \label{theorem:distributed-transshipment} \textnormal{(Distributed transshipment). }
  Let $G$ be a weighted graph with weights in $[1, n^{O(1)}]$ and let $\eps > 0$ be a parameter. There exists an algorithm that computes $(1+\eps)$-approximate transshipment on $G$ in $\eps^{-2} \cdot \exp(O(\log n \cdot \log \log n)^{3/4})$ Minor-Aggregation rounds.
\end{restatable}

Combining the Minor-Aggregation model result with the simulation of \Cref{theorem:informal-dmam-simulation}, we immediately obtain the following set of results:
\begin{corollary}[SSSP and transshipment in CONGEST]
  Let $G$ be a weighted graph with weights in $[1, n^{O(1)}]$. We can solve both the $\left(1+\frac{1}{n^{o(1)}}\right)$-approximate SSSP and $\left(1+\frac{1}{n^{o(1)}}\right)$-approximate transshipment in $\opt(G) \cdot n^{o(1)}$ CONGEST rounds, where $\opt_G$ is the runtime of the fastest possible correct algorithm that works on $G$. Moveover, this algorithm is guaranteed to complete in $D \cdot n^{o(1)}$ CONGEST rounds if $G$ is a planar, genus-bounded, treewidth-bounded, minor-free graph, or an $n^{-o(1)}$-expander.
\end{corollary}

The above set of results all rely on the distributed evaluation of $\ell_1$-oblivious routing, which we state below for reference.
\begin{restatable}{theorem}{theoremDistributedEvaluationRouting}
  \label{theorem:distributed-evaluation-routing}\textnormal{(Distributed oblivious routing evaluation).}
  Let $G = (V, E)$ be a weighted graph with weights in $[1, n^{O(1)}]$. There exists an $\exp(O(\log n \cdot \log \log n)^{3/4})$-approximate $\ell_1$-oblivious routing $R$ such that we can perform the following computations in $\exp(O(\log n \cdot \log \log n)^{3/4})$ rounds of Minor-Aggregation. Given any distributedly stored node vector $d \in \R^V$ and edge vector $c \in \R^V$, we can evaluate and distributedly store $R d \in \R^{\vec{E}}$ and $R^T c \in \R^V$.
\end{restatable}

\subsection{Related Work}
Techniques from continuous optimization have brought breakthroughs for many problems in combinatorial optimization. Sherman \cite{She13} obtained an approximate max-flow algorithm in undirected graphs with time complexity $O(m^{1+o(1)}\eps^{-2})$ by constructing an $n^{o(1)}$-congestion approximator \cite{Mad10} in nearly linear time. Peng~\cite{Pen16} improved the run-time to $O(m\mathrm{polylog}(n))$ by solving the approximate max-flow and constructing the cut-based hierarchical tree decompositions \cite{RST14} recursively, and using ultra-sparsifiers \cite{KMP10,ST14} to reduce the graph size. Sherman \cite{She17b} designed a generalized preconditioner and applied that to produce an approximation algorithm for uncapacitated min-cost flow on undirected graphs with time complexity $O(m^{1+o(1)}\eps^{-2})$. Kelner et al. \cite{KLOS14} solved the approximate maximum concurrent multicommodity flow problem with $k$ commodities in time $O(m^{1+o(1)}k^2\eps^{-2})$ by an $O(m^{o(1)})$-competitive oblivious routing scheme. Sherman \cite{She17} improved this time complexity to $O(mk\eps^{-1})$ via the proposed area-convex regularization. In the distributed settings, Ghaffari et al. \cite{GKKLP15} proposed an $(1+o(1))$-approximate algorithm that takes $(\sqrt{n}+D)n^{o(1)}$ rounds for max-flow in undirected weighted networks, where $D$ is the hop-diameter of the communication network. 


For the directed min-cost flow problem, Daitch and Spielman \cite{DS08} provided an algorithm that runs in time $\widetilde{O}(m^{3/2}\log^2{U})$ by solving the linear equations efficiently; this time complexity was improved to $\widetilde{O}(m\sqrt{n}\log^2{U})$ by Lee and Sidford \cite{LS14} by virtue of the new algorithm for solving linear programs. In terms of transshipment which is a generalization of the min-cost flow problem, Sherman \cite{She17b} proposed an $(1+\varepsilon)$-approximate algorithm with time complexity $m^{1+o(1)}\varepsilon^{-2}$ for uncapacitated and undirected min-cost flow problem using a generalized preconditioner obtained by $\ell_1$ embedding and a hierarchical routing scheme in $\ell_1$; Li \cite{Li20} proposed a parallel algorithm with $m\mathrm{polylog}(n)\varepsilon^{-2}$ work and $\mathrm{polylog}(n)\varepsilon^{-2}$ time for the transshipment problem in \cite{She17b} by utilizing the multiplicative weight update method to solve linear programs. A crucial ingredient tying together these approaches is a property of transshipment that solvers that return an approximate dual solution can be \emph{boosted} to an $(1 + \eps)$-approximate solver~\cite{zuzic2021simple}.

In the CONGEST model, substantial progress has been made on the SSSP problem in the past decade. For undirected graphs, Lenzen and Patt-Shamir \cite{LP13} proposed an $O(\varepsilon^{-1}\log{\varepsilon^{-1}})$-approximation algorithm that runs in $\widetilde{O}(n^{1/2+\varepsilon}+D)$ rounds for the weighted SSSP problem. This approximation factor was improved by \cite{Nan14}, in which Nanongkai presented an $\widetilde{O}(\sqrt{n}D^{1/4}+D)$-round $(1+o(1))$-approximation algorithm. Elkin and Neiman \cite{EN16} gave an $(1+\varepsilon)$-approximation algorithm that runs in $(\sqrt{n}+D)2^{\widetilde{O}(\sqrt{\log{n}})}$ rounds using their hopsets. Henzinger et al. \cite{HKN16} proposed a deterministic $(1+o(1))$-approximation algorithm that takes $n^{1/2+o(1)}+D^{1+o(1)}$ rounds. Haeupler and Li \cite{HL18} proposed the first approximation algorithm for SSSP that adjusts to the topology of the communication network, and gave a $\widetilde{O}(\shortcutquality(G)n^{o(1)})$-round algorithm with an approximation factor of  $n^{o(1)}$, where $\shortcutquality(G)$ is the quality of the best shortcut that can be constructed efficiently for the given topology. Ghaffari and Li \cite{GL18b} designed distributed algorithms that run in $\tau_{mix}\cdot 2^{O(\sqrt{\log{n}})}$ rounds for $(1+\eps)$-approximate SSSP and transshipment, where $\tau_{mix}$ is the mixing time of the graph. More generally, \cite{GL18b} proved that each parallel algorithm with efficient work that takes $\tau$ rounds can be simulated on any (expander) graph in $\tau_{mix}\cdot 2^{O(\sqrt{\log{n}})}\cdot \tau$ rounds in the CONGEST model. 
Besides these approximate algorithms, Elkin \cite{Elk17} devised an exact SSSP algorithm for undirected graph that requires $O((n\log{n})^{5/6})$ rounds when $D=O(\sqrt{n\log{n}})$, and $O((n\log{n})^{2/3}D^{1/3})$ rounds for larger $D$. 

For the directed SSSP problem, Ghaffari and Li \cite{GL18} presented an $\widetilde{O}(n^{3/4}D^{1/4})$-round algorithm, and an improved algorithm for large values of $D$ with round complexity $\widetilde{O}(n^{3/4+o(1)}+\min\{n^{3/4}D^{1/6}, D^{6/7}\}+D)$. Forster and Nanongkai \cite{FN18} gave two randomized algorithms that take $\widetilde{O}(\sqrt{nD})$ and $\widetilde{O}(\sqrt{n}D^{1/4}+n^{3/5}+D)$ rounds respectively, and obtained an $(1+\eps)$-approximation algorithm with  $\widetilde{O}((\sqrt{n}D^{1/4}+D)\varepsilon^{-1})$ rounds. Chechik and Mukhtar \cite{CM20} proposed a randomized algorithm that takes $\widetilde{O}(\sqrt{n}D^{1/4}+D)$ rounds. Cao et al. \cite{CFR21} provided an $(1+\varepsilon)$-approximation algorithm with round complexity $\widetilde{O}((\sqrt{n}+D+n^{2/5+o(1)}D^{2/5})\varepsilon^{-2})$. Censor-Hillel et al. \cite{CLP21} gave an exact SSSP algorithm with $(n^{7/6}/m^{1/2}+n^2/m)\cdot\tau_{mix}\cdot 2^{O(\sqrt{\log{n}})}$ rounds in the CONGEST model. 

Besides CONGEST model, SSSP problem was also studied in the Congested Clique model~\cite{Nan14,HKN16, CHDKL20}, Broadcast CONGEST model~\cite{BKKL17}, Broadcast Congested Clique model \cite{BKKL17}, and Hybrid model~\cite{AHKSS20,FHS20,KS20,CLP21}.

In addition, the all-pairs shortest paths (APSP) problem was also studied in various distributed models~\cite{HW12,CHKKLPS15,HP15,Le16,Elk17,HNS17,ARKP18,CHLT18,PR18,AR19,BN19,AHKSS20,AR20,CHDKL20,DP20,KS20,CLP21,DFKL21}. 

\section{Preliminaries}\label{sec:prelims}

\textbf{Graph Notation} Let $G = (V, E)$ be a simple undirected graph. We denote with $n := |V|$ the number of nodes, with $m := |E|$ the number of edges, and with $D$ the hop-diameter of $G$.
It is often convenient to direct $E$ consistently. For simplicity and without loss of generality, we assume that the vertices $V = \{ v_1, \ldots, v_n \}$ are numbered from $1$ to $n$, and we define $\vec{E} = \{ (v_i,v_j) \mid (v_i, v_j) \in E, i < j\}$. We identify $E$ and $\vec{E}$ by the obvious bijection.
We denote with $B \in \{-1, 0, 1\}^{V \times \vec{E}}$ the node-edge incidence matrix of $G$, which for any $e = (s,t) \in \vec{E}$ assigns $B_{s,e} = 1$, $B_{t,e} = -1$, and $B_{u,e} = 0$ for all other $u \in V$.
In this paper, we typically assume the graph is weighted and the weights are polynomially bounded. For this, a weight or length function $w$ assigns each edge $e \in E$ a weight $w_e \in [1,n^{O(1)}]$. The weight function can also be interpreted as 
a diagonal \emph{weight matrix} $W \in [0,n^{O(1)}]^{E \times E}$ which assigns $W_{e,e} = w_e \ge 1$ for any $e \in \vec{E}$ (and $0$ on all off-diagonal entries).

\textbf{Flows and Transshipment}\label{sec:prelim-transshipment}
A \emph{demand} is a $d \in \R^{V}$. We say a demand is \emph{proper} if $\one^T d=0$. A \emph{flow} is a vector $f \in \R^{\vec{E}}$. A flow $f$ \emph{routes} demand $d$ if $Bf = d$. It is easy to see only proper demands are routed by flows. The cost $W(f)$ of a flow $f$ is $\norm{Wf}_1$.
For a weighted graph $G$ and a given proper demand $d$ the \emph{transshipment problem} asks to find a flow $f_d^*$ of minimum cost among flows that route $d$. When the underlying graph is clear from the context, we let $\norm{d}_\opt := W(f_d^*)$ denote the cost of the optimal flow for routing demand $d$. For any $\alpha \geq 1$, we say a flow $f$ is $\alpha$-\emph{approximate} if $W(f) \leq \alpha \cdot \norm{d}_\opt$ (such a flow does not necessarily need to route $d$).
The transshipment problem naturally admits the following convex-programming formulation:
\begin{align}
  \min~ \norm{Wf}_1 : Bf = d, \label{eq:transshipmentLP}
\end{align}
and its dual:
\begin{align}
  \max~ d^{\top} \potential : \norm{W^{-1} B^\top \phi}_{\infty} \leq 1. \label{eq:transshipmentDualLP}
\end{align}
The entries in the vector $\phi \in \R^{V}$ are generally referred to as vertex \emph{potentials}.

\textbf{Distributed Algorithm (i.e., CONGEST Model).} The distributed algorithms designed in this paper are message-passing algorithms following the standard CONGEST~\cite{peleg2000distributed} model of distributed computing. A network is given as an undirected graph $G = (V, E)$ in which nodes are individual computational units (i.e., have private memory and do their own computations). Communication between the nodes occurs in synchronous rounds. In each round, each pair of nodes adjacent in $G$ exchanges an $O(\log n)$-bit message. Nodes perform arbitrary computation between rounds. Initially, nodes only know their unique $O(\log n)$-bit ID and the IDs of adjacent nodes as well as the weights of incident edges, for problems like SSSP with weights as input. 

\textbf{Asymptotic Notation.}
We use $\tilde{O}$ to hide polylogarithmic factors in $n$, i.e., $\tilde{O}(1) = \polylog n $ and $\tilde{O}(f(n,D)) = O(f(n,D) \cdot  \polylog n)$. We use the term \emph{with high probability} to denote success probabilities of at least $1 - n^{-C}$ where $C > 0$ is a constant that can be chosen arbitrarily large.


\subsection{Low-Congestion Shortcut Framework and Part-wise Aggregation}

Our universally-optimal distributed algorithms build on the \emph{low-congestion shortcut framework}~\cite{GH16,haeupler2016low,haeupler2016near,KKOI19,HWZ21} which identifies the \emph{part-wise aggregation task} as the crucial communication task in many optimization problems. The framework also introduced shortcuts as a near-optimal way of solving this communication problem. We summarize the facts needed in this paper and refer for more details to \cite{HWZ21}.

\textbf{Aggregations.} Let $\bigoplus$ be some function that combines two $O(\log n)$-bit messages into one (e.g., sum or max). We call such an $\bigoplus$ an \emph{aggregation operator}. Given messages $x_1, x_2, \ldots, x_k$, their \emph{aggregate} $\bigoplus_{i=1}^k x_i$ is the resulting message after iteratively taking two arbitrary messages $m', m''$, deleting them, and inserting $m' \bigoplus m''$ into the sequence until a single message remains (e.g., the sum-aggregation of $x_1, \ldots, x_k$ is $x_1+\ldots+x_k$). Throughout this paper, $\bigoplus$ will be commutative and associative, making the value $\bigoplus_{i\in I} x_i$ unique and well-defined even if the order under which the messages are aggregated changes. However, it is often useful to consider more general aggregations, allowing us to use any \emph{mergeable $O(\log n)$-bit sketches}~\cite{agarwal2013mergeable} as an aggregation operator (examples of such operators include approximate heavy hitters, quantile estimation, frequency estimations, random sampling, etc.).

\textbf{Part-wise aggregation (PA).} The PA task is a central task used to solve many distributed operations problems. A formal definition follows.
\begin{definition}[Part-wise aggregation (PA) task]
  \label{def:pa-task}
  Suppose that the nodes $V$ of a graph $G = (V, E)$ are subdivided into a set of \textbf{connected} and \textbf{node-disjoint} parts $P_1, \ldots, P_k$. Initially, each node $v$ chooses a private $O(\log n)$-bit input $x_v$ and the task is to compute, for each part $P_i$, the aggregate over all private inputs belonging to that part, i.e., $\bigoplus_{v \in P_i} x_v$ (all nodes learn the same value). 
\end{definition}

Low-congestion shortcuts near-optimally solve the PA task. Specifically, \cite{GH16, HWZ21} define a function $\shortcutquality(\cdot)$ which assigns any undirected and unweighted graph $G$ a positive integer $\shortcutquality(G)$ which characterizes (both as an upper and lower bound) the distributed complexity of solving PA. The lower bound is based on the network coding gaps of \cite{HWZ20, HWZ21}, while the upper bound is based on the very recent hop-constrained oblivious routings and efficient hop-constrained expander decompositions~\cite{ghaffari2021hop, GHR21}.

\begin{theorem}\label{thm:pa-simulation}
  Suppose $A$ is a distributed CONGEST algorithm solving the part-wise aggregation task. For any $G$, the running time of $A$ on $G$ is at least $\tilde{\Omega}(\shortcutquality(G))$ rounds~\cite{HWZ21}. Moreover, there is a randomized CONGEST algorithm that solves any PA instance in $\poly(\shortcutquality(G)) \cdot n^{o(1)}$ rounds~\cite{GHR21}. Specifically, if $\shortcutquality(G) \le n^{o(1)}$, then we can solve PA in $n^{o(1)}$ rounds.
\end{theorem}
We remark that on some special graph classes like excluded-minor graphs there exist algorithms that lose only $\tilde{O}(1)$ factors on the upper bound and can be made deterministic.~\cite{haeupler2016low,ghaffari2020low}

What makes part-wise aggregation and shortcuts so powerful is that they relate to many important (non-local) algorithmic problems, including, min-cut, MST, connectivity, etc. In particular, $\shortcutquality(G)$ is a lower bound on many important optimization problems, including all shortest-path problems considered in this paper.

\begin{theorem}[\cite{HWZ21}]\label{thm:lower-bound-shortcuts}
Suppose $A$ is a (correct) distributed algorithm for SSSP or transshipment (or MST, min-cut, max-flow, etc.) with a sub-polynomial approximation ratio. For any $G$, the running time of $A$ on $G$ is at least $\tilde{\Omega}(\shortcutquality(G))$ rounds.
\end{theorem}

On the other hand, many problems can be solved with a small number of part-wise aggregations, giving universally-optimal algorithms with matching $\tilde{\Theta}(\shortcutquality(G))$ lower and upper bounds on the round complexity (when $\shortcutquality(G) \le n^{o(1)}$).
The contribution of this paper is to show that this is true for approximate shortest-path problems.

We note that current solution to PA in $\poly(\shortcutquality(G)) \cdot n^{o(1)}$ only allows for $n^{o(1)}$-competitive universal optimality when there exists a $n^{o(1)}$-round solution. This is typically sufficient as many networks for interest have $\shortcutquality(G) \le n^{o(1)}$. Moreover, a future results might improve the PA upper bound to $\tilde{O}(\shortcutquality(G))$, in which case we would always get polylog-competitive universally-optimal algorithms.

Furthermore, on many special graph classes not only does the low-congestion shortcut framework guarantee that any universally-optimal algorithm is competitive with the fastest correct one, it also provides concrete upper bounds on this runtime. For example, when run on a planar graph, we know that $\shortcutquality(G) = \tilde{O}(D)$ and we know how to solve the PA problem in $\tilde{O}(D)$ rounds, giving us a useful guarantee on the runtime on any algorithm that uses PA as a subroutine or any universally-optimal algorithm (e.g., the one presented in this paper). We provide a compiled list of concrete shortcut qualities and PA runtimes for specific graph classes.

\begin{theorem}[Shortcut quality on special graphs]\label{theorem:shortcut-on-specific-graphs}
  Let $G$ be a (undirected and unweighted) graph and let $D := \diameter(G)$. The following bounds hold:
  \begin{itemize}
  \item For all graphs $G$ we have $\shortcutquality(G) \le \tilde{O}(\sqrt{n} + D)$. PA can be solved in deterministic $\tilde{O}(\sqrt{n} + D)$ rounds.~\cite{GH16}
  \item When $G$ is planar, then $\shortcutquality(G) \le \tilde{O}(D)$. PA can be solved in deterministic $\tilde{O}(D)$ rounds.~\cite{GH16}
  \item When $G$ has excluded minor, then $\shortcutquality(G) \le \tilde{O}(D)$ (the hidden constants depend on the excluded minor). The PA problem can be solved in deterministic $\tilde{O}(D)$ rounds. This generalizes the result for planar graphs. The same result holds even when $G$ excludes $\tilde{O}(1)$-dense minors.~\cite{ghaffari2020low}
  \item When $G$ has treewidth at most $k$, then $\shortcutquality(G) \le \tilde{O}(kD)$. PA can be solved in randomized $\tilde{O}(kD)$.~\cite{haeupler2016near,haeupler2018round}
  \item When $G$ has genus at most $g \ge 1$, then $\shortcutquality(G) \le \tilde{O}(\sqrt{g}D)$. PA can be solved in randomized $\tilde{O}(\sqrt{g}D)$ rounds.~\cite{haeupler2016near}
  \item When $G$ is an $n^{-o(1)}$-expander, then $\shortcutquality(G) \le n^{o(1)}$. PA can be solved in randomized $n^{o(1)}$ rounds.~\cite{GL18b}
  \item When $G$ has hop-diameter $D$, then $\shortcutquality(G) \le \tilde{O}(n^{(D - 2) / (2D - 2)} + D)$. For $3 \le D = O(1)$, PA can be solved in randomized $\tilde{O}(n^{(D - 2) / (2D - 2)})$ rounds.~\cite{kogan2021low, KKOI19}
  \end{itemize}
\end{theorem}


\section{Graph-Based $\ell_1$-Oblivious Routing via LDDs}\label{sec:oblivious-routing-general}
In this section, we present our graph-based $n^{o(1)}$-approximate $\ell_1$-oblivious routing $R$. We first define them in \Cref{sec:def-oblivious-routing}. \Cref{sec:oblivious-routing-construction} then presents an existential, model-oblivious, construction of $R$ and proves its approximation guarantees. \Cref{sec:linear-algebra-oblivious-routing} presents a linear-algebraic interpretation of the constructed oblivious routing $R$.

\subsection{Definition}\label{sec:def-oblivious-routing}

For any $k$ and $\alpha > 1$, Sherman~\cite{She17b} defined an $\alpha$-approximate (linear) $\ell_1$-preconditioner for a weighted graph $G$ as a $k \times n$ matrix $P$, such that, for any proper demand $d$ it holds that $$\norm{d}_\opt \le \norm{P d}_1 \le \alpha \norm{d}_\opt.$$ Any such $\alpha$-preconditioner which can be computed in $\tau$ rounds can be used to give a $\tilde{O}(\eps^{-2} \alpha^2 \tau)$ algorithm for transshipment~\cite{She17}. We are constructing a strictly stronger object that is required to output a flow that routes any demand with $\alpha$-approximate cost.

\begin{definition}
For a graph $G=(V,E)$, a matrix $R \in \R^{\vec{E} \times V}$ is a (linear) oblivious routing if for any proper demand $d \in \R^V$ the flow $Rd$ routes $d$, i.e., $B R d = d$ for all $d \in \R^V$ with $\sum_{v \in V} d_v = 0$.
\end{definition}

\begin{definition}
A routing $R$ for a graph $G=(V,E)$ with weight matrix $W$ is an $\alpha$-approximate (linear) $\ell_1$-oblivious routing if for any proper demand $d$ the cost of the flow $f_{R,d}=Rd$ is at most $\alpha \norm{d}_\opt$, i.e., $W(f_{R,d}) = \norm{W R d}_1 \leq\alpha\norm{d}_\opt$. 
\end{definition}

\begin{corollary}
If $R$ is an $\alpha$-approximate $\ell_1$-oblivious routing for a graph $G$ with weight matrix $W$ then $P = W R$ is an $\alpha$-approximate (linear) $\ell_1$-preconditioner for $G$.
\end{corollary}
\begin{proof}
  $$\norm{d}_\opt = \min_{f: Bf=d} \norm{Wf}_1 \leq \norm{W f_{R,d}}_1 = \norm{W R d}_1 \leq\alpha\norm{d}_\opt. \qedhere$$
\end{proof}

We remark that Li~\cite{Li20}, like Sherman~\cite{She17}, constructs a sparse $\alpha$-approximate (linear) $\ell_1$-preconditioner $P$ but at times calls such a matrix $P$ an $\ell_1$-oblivious routing even though $P d$ only produces a vector with the right norm and not a flow/routing. To our knowledge, this paper gives the first matrix $R$ that can be evaluated in almost linear time and is an $\alpha$-approximate $\ell_1$-oblivious routing with sub-polynomial $\alpha$.

\subsection{Construction}\label{sec:oblivious-routing-construction}

The main graph-theoretic tool we use to construct our $\ell_1$-oblivious routing is the well-studied low-diameter decomposition or LDD~\cite{awerbuch1985complexity,awerbuch1990sparse,bartal1996probabilistic}. Informally, an LDD decomposes a graph into disjoint node partitions such that each pair of close nodes has a large probability of ending up in the same part. A formal definition follows.
\begin{definition}[Low-diameter decomposition]
  For a weighted graph $G = (V, E)$, a low-diameter decomposition (LDD) $\mathcal{P}$ of radius $\rho$ and quality $\alpha \ge 1$ is a probability distribution over node disjoint partitions of $V$ into (connected) components $S_1 \subseteq V, \ldots, S_k \subseteq V$ along with centers $c_1 \in S_1, \ldots, c_k \in S_k$ such that:
  \begin{enumerate}
  \item For each $i$, the center $c_i$ is within distance $\rho$ of every other node in the induced subgraph $G[S_i]$, w.h.p.
  \item For every two vertices $x, y \in V$, if $d$ is the distance between them in the original graph, then the probability that they do not belong to the same part $S_i \supseteq \{u, v\}$ is at most $\alpha \cdot \frac{d}{\rho}$.
  \end{enumerate}
\end{definition}

Our $\ell_1$-oblivious routing algorithm based on LDDs is given in \Cref{algo:oblivious-routing-construction}.

\begin{algorithm}[H]
  \caption{$\ell_1$-oblivious routing}
  \label{algo:oblivious-routing-construction}
  \begin{enumerate}\itemsep0em
  \item Let $\alpha := \exp(O(\sqrt{\log n \cdot \log \log n}))$ (representing the LDD quality).
  \item Let $\rho := \exp(O(\log n \cdot \log \log n)^{3/4})$ (representing the LDD radius).
  \item Let $d_0 \in \R^V$ by any demand vector.
  \item For $i = 1, 2, \ldots, \imax := O(\log n)^{1/4}$ repeat the following:
    \begin{enumerate}
    \item Sample $g := O(\log n) \cdot \frac{\rho}{\alpha}$ LDDs $\calP_1^i, \calP_2^i, \ldots, \calP_g^i$ with radius $\rho^i$ and quality $\alpha$.
    \item A node $v$ sends $\frac{1}{g} d_{i-1}(v)$ to the center of its component in each $\calP_j^i$ along any path of length at most $\rho^i$, for each $j \in \{1, \ldots, g\}$. This constructs the next-step demand $d_i$.
    \end{enumerate}
    
  \item Compute an arbitrary spanning tree $T$ of $G$ and choose an arbitrary root $r$. Each node $v$ sends its remaining flow $d_{\imax}(v)$ to $r$ along $T$. \label{algo-step:spanning-tree}

  \end{enumerate}
\end{algorithm}
\textbf{Choosing the constants.} Each hidden constant in the $O$-notation of \Cref{algo:oblivious-routing-construction} can be replaced by a universal constant. Specifically, the $O$-constant in the definition of $\alpha$ is inherited from prior work~\cite{HL18} and explained in \Cref{theorem:distributed-LDD-sampling}. Constants in the definition of $\rho$ and $\imax$ are arbitrary as long as $\rho^{\imax} \ge n^{C+1}$, where $n^C$ is the largest edge weight in the graph, with $C = O(1)$ as we assumed they are polynomially bounded. The constant in the definition of $g$ makes the algorithm succeed with high probability, hence making it a sufficiently large constant drives the success probability to at least $1 - n^{-C'}$ for any chosen $C' = O(1)$.

\begin{theorem}\label{theorem:oblivious-routing-approximation}
  With high probability, \Cref{algo:oblivious-routing-construction} produces an $\exp(O(\log n \cdot \log \log n)^{3/4})$-approximate $\ell_1$-oblivious routing.
\end{theorem}
\textbf{Remark.} If one would use (optimal) LDDs of quality $O(\log n)$, the construction described would yield an $\exp(O(\sqrt{\log n \cdot \log \log n}))$-approximate $\ell_1$-oblivious routing.

\textbf{A guided tour of the analysis.} We define a \emph{pair-demand} $d_{s, t} \in \R^V$ as $d_{s, t}(x) := \1{x = s} - \1{x = t}$, i.e., requesting a unit flow from $s$ to $t$. Due to linearity of our routing, it is sufficient to prove that the routing offers a good approximation only with respect to all pair demands $d_{s, t}$ in order to prove it does the same for all (non-pair) demands. Therefore, the assumption that $d$ is a pair-demand is without loss of generality. This greatly simplifies the analysis.

We analyze how the optimal solution changes when routed along a single LDD $\calP$ (of quality approximately $2^{\sqrt{\log n}}$). Since $d = d_{s, t}$ the optimal solution $\norm{d}_\opt$ is initially equal to the distance between between $s$ and $t$, namely $\norm{d}_\opt = \ell := \dist_G(s, t)$. After routing along a single LDD with radius $\rho$ and routing each $d(x)$ to the center of the component of $x$, we pay a cost of at most $2\rho$ to route the demand and obtain a new (residual) demand $d' \in \R^V$ (which is supported at centers of components of $\calP$). Note that if $s$ and $t$ are in the same component of $\calP$, then $d' = \vec{0}$; if they are in different components we have $\norm{d'}_\opt \le \norm{d}_\opt + 2 \rho$ since $s$ and $t$ are both moved by distance at most $\rho$. However, in expectation, the increase is $\E[\norm{d'}_\opt - \norm{d}_\opt] \le \Pr[\text{s, t in different components}] \cdot O(\rho) = 2^{\sqrt{\log n}} \cdot \frac{\ell}{\rho} \cdot O(\rho) \le 2^{\sqrt{\log n}} \cdot \norm{d}_\opt$. In other words, the optimal solution increases in expectation by a manageable $2^{\sqrt{\log n}}$ factor. A naive way to use this result would be to a single hierarchy of LDD decompositions (i.e., find an LDD of radius $\rho = 2^{(\log n)^{3/4}}$, contract, and repeat). We showed this loses a $2^{\sqrt{\log n}}$ multiplicative factor in the value of the optimal solution in each level and, if done for $(\log n)^{1/4}$ levels, ultimately loses a $2^{(\log n)^{3/4}} = n^{o(1)}$ factor. This constructs an $n^{o(1)}$-approximate routing for any demand \emph{in expectation}. It essentially corresponds to a low-quality tree embedding---an analogous (but better!) routing with a $O(\log n)$-approximation guarantee in expectation would be to simply sample an FRT tree~\cite{fakcharoenphol2004tight}.

This, however, completely ignores the issue of concentration---we require our single constructed routing to be good with respect to \emph{all demands}, and not just a single demand in expectation. This can typically be done by repeating the same process a sufficient number of times and taking the average until it works with high probability. For example, would need to repeat the above LDD-hierarchy routing at least $\poly(n)$ times until the averaged-out process succeeds for \emph{each demand}. We address this by achieving concentration for each level of the LDD hierarchy. Most of the conceptual heavy-lifting comes from proving that \Cref{algo:oblivious-routing-construction} works for all demands with high probability, since we have showed that designing a demand that only works in expectation is straightforward. To this end, our idea is to gradually increase the LDD radii as $\rho := 2^{(\log n)^{3/4}}, \rho^2, \rho^3, \ldots, \rho^{\imax} = \poly(n)$ and showing concentration in-between each step. As we shown before, the optimal solution \emph{in expectation} increases in each step by a manageable $2^{\sqrt{\log n}}$ factor, hence the total blow-up after $\imax \le (\log n)^{1/4}$ step is still $n^{o(1)}$, while the LDDs start consuming the entire graph, implying we are done (it is not hard to see why). To maintain concentration, in each step we sample $\tilde{O}(\rho) \approx 2^{(\log n)^{3/4}}$ LDDs to guarantee the optimal solution does not blow up by more than a $n^{o(1)}$ factor in each step. The crux of the analysis is in showing this averaging will keep the optimal solution small in every step and with respect to every demand.

It is fairly straightforward to show concentration when $i = 1$: the optimal solution is good in expectation and increases by at most an additive $\rho$ factor, hence repeating it for $\tilde{O}(\rho)$ steps guarantees the result with high probability---this is a standard Chernoff bound result, we need to take the average over $\tilde{O}(M / \mu)$ independent random variables that have expectation $\mu$ and take up values in the range $[0, M]$ with probability $1$. The issue, however, arises when $i > 1$: the optimal solution can grow by a factor of $\rho^i \gg \norm{d}_\opt$. Naively, this tells us we need to repeat the process for $\rho^i / \norm{d}_\opt$ times, a value which can be often as large as $\rho^i$. The trick, however, is to ``artificially'' increase the value of the optimal solution after routing it along an LDD of large radius. Specifically, after routing the demand along an LDD of radius $\rho$, we will increase the value of the optimal solution by an additive $\rho$. This does not influence the \emph{newly-increased optimal solution} significantly---the expectation even remains the same up to constant factors. Notably, this increase helps to prove the newly-increased solution concentrates: in step $i$, the optimal solution increases by an additive $\rho^i$ factor, but is of size at least $\rho^{i-1}$ after the previous step, hence repeating it $\tilde{O}(\rho^i / \rho^{i-1}) = n^{o(1)}$ times is sufficient to prove concentration (i.e., the newly-increased optimal solution does not blow up with respect to all demands). In the following formal proof, we simplify much of this conceptually complicated reasoning by introducing a potential $\phi_i := \norm{d_i}_\opt + \norm{d_i} \cdot \rho^i$ which intuitively corresponds to the newly-increased optimal solution.

\subsection{Proving the approximation guarantees of \Cref{algo:oblivious-routing-construction}}

This section is dedicated to proving \Cref{theorem:oblivious-routing-approximation}. The main technical insight that greatly simplifies the analysis is the following definition of a potential, which accounts for both the value of the optimal solution and the remaining $\ell_1$ mass into account. Note that the subscript $i$ corresponds to the $i^{th}$ step of \Cref{algo:oblivious-routing-construction}.

\begin{definition}[Potential]
  $\Phi_i := \norm{d_i}_\opt + \norm{d_i}_1 \cdot \rho^i$.
\end{definition}

Our ultimate goal is to prove the potential increases only by a $n^{o(1)}$ factor over all $\imax$ many steps; this can be shown to directly imply \Cref{theorem:oblivious-routing-approximation}. We prove this in several steps. First, we show that in each step $i$ the potential increases only by a multiplicative $O(\alpha) = n^{o(1)}$, but only \emph{in expectation} and when the demand $d_{i-1}$ is fixed to be a \emph{pair-demand} $d_{i-1} = d_{s, t}$. A \textbf{pair-demand} $d_{s, t} \in \R^{V}$ is a demand of the form $d_{s, t}(x) := \1{x = s} - \1{x = t}$, i.e., routing a unit flow from $s$ to $t$. Second, we show the same claim with high probability instead of in expectation. Third, we show the claim for all demands (not only pair-demands) with high probability. Together, they imply the result.

We start by showing the result in expectation and for pair-demands.
\begin{lemma}[Expectation analysis]\label{lemma:expectation-analysis}
  For any $1 \le i \le \imax$ and any two nodes $s, t \in V$ the following holds. Fix $d_{i-1}$ to be a pair-demand $d_{s, t}$ (with potential $\Phi_{i-1} = d_G(s, t) + 2 \rho^{i-1}$) and run the $i^{th}$ step of \Cref{algo:oblivious-routing-construction}. This induces a new random variable $\Phi_i$. We always have that $\E[\max\{\Phi_i, \Phi_{i-1}\}] \le 6 \alpha \cdot \Phi_{i-1}$ over the random choices of the $i^{th}$ step. 
\end{lemma}
\begin{proof}
  Suppose that the (weighted) distance in $G$ between $s$ and $t$ is $\ell$. Since we fixed $d_{i-1} = d_{s, t}$, then $\Phi_{i-1}=\ell + 2 \rho^{i-1} \ge \ell$.

  We first analyze a \emph{single} randomly sampled LDD with potential $\Phi_i^{\text{single}}$. We define a random variable $M_i^{\text{single}} := \max\{\Phi_i^{\text{single}}, \Phi_{i-1}\}$ which tracks the (single) potential without ever decreasing it.

  If both endpoints $s, t$ are in the same component, this demand cancels out and $\Phi_i^{\text{single}} = 0$, which in turn gives us $M_i^{\text{single}} = \Phi_{i-1}$. On the other hand, if they end up in separate components (which happens with probability at most $\min\{ 1, \alpha \cdot \frac{\ell}{\rho^i} \}$), we route them to separate component centers (that are at most $\rho^i$ away), which gives the following bound on the potential
  \begin{align*}
    \Phi_i^{\text{single}} = \norm{d_i}_{\opt} + \norm{d_i}_1 \cdot \rho^i \le (\ell + 2 \rho^i) + 2 \cdot \rho^i = \ell + 4 \cdot \rho^i .
  \end{align*}
  Note that the expectation does not change whether we take a single LDD or an average over $g$ LDDs, hence we can drop the superscript on ${\Phi_i}^{\text{single}}, M_i^{\text{single}}$ and just write $\Phi_i$ and $M_i$. In expectation, we get:
  \begin{align*}
    \E[ M_i ] & \le \Phi_{i-1} + \min\left\{ \alpha \cdot \frac{\ell}{\rho^i}, 1 \right\} \cdot (\ell + 4 \rho^i) .
  \end{align*}
  We now consider two cases. First, if $\ell \le \rho^i$, we get
  \begin{align*}
    \E[M_i] \le \Phi_{i-1} + \alpha \cdot \frac{\ell}{\rho^i} \cdot 5\rho^i = \Phi_{i-1} + 5 \alpha \cdot \ell \le \Phi_{i-1} + 5 \alpha \cdot \Phi_{i-1} \le 6\alpha \cdot \Phi_{i-1} .
  \end{align*}
  Otherwise, if $\ell > \rho^i$, we get
  \begin{align*}
    \E[M_i] \le \Phi_{i-1} + \ell + 4 \rho^i \le \Phi_{i-1} + 5 \ell \le 6 \Phi_{i-1} .
  \end{align*}
  Therefore, it always holds that $\E[\Phi_i] \le \E[M_i] \le 6\alpha \cdot \Phi_{i-1}$.
\end{proof}

Next, in \Cref{lemma:concentration-analysis}, we show that the potential does not blow up significantly \emph{with high probability}, instead of just being bounded in expectation. For this we use the following standard Chernoff bound:

\begin{lemma}[Chernoff Bound] \label{lem:ChernoffBound} Let $S = \sum_i X_i$ be a sum of non-negative independent random variables upper bounded by $Z$, i.e., $0 \le X_i \le Z$ with probability $1$. For any $t \ge 2 \E[S]$ it holds that $\Pr[S \ge t] \le \exp(- \frac{1}{3} \cdot t / Z)$.
\end{lemma}

\begin{lemma}[Concentration analysis]\label{lemma:concentration-analysis}
  For any $1 \le i \le \imax$ and any two nodes $s, t \in V$ the following holds. Fix $d_{i-1}$ to be a pair-demand $d_{s, t}$ (with potential $\Phi_{i-1} = d_G(s, t) + 2 \rho^{i-1}$) and run the $i^{th}$ step of \Cref{algo:oblivious-routing-construction}. This induces a new random variable $\Phi_i$. With high probability, $\Phi_i \le 13 \alpha \cdot \Phi_{i-1}$ over the random choices of the $i^{th}$ step.
\end{lemma}
\begin{proof}
  We first analyze a \emph{single} randomly sampled LDD with potential $\Phi_i^{\text{single}}$. We define a random variable $M_i^{\text{single}} := \max\{\Phi_i^{\text{single}}, \Phi_{i-1}\}$ which tracks the (single) potential without decreasing it.

   Let $\Delta_i^{\text{single}} := M_i^{\text{single}} - \Phi_{i-1} = \max\{0, \Phi_i^{\text{single}} - \Phi_{i-1}\}$ be the random variable denoting the change in potential between steps $i-1 \to i$ after sampling a single LDD (or $0$ if the change is negative).

    Naturally, the final difference in potentials is $\Phi_i - \Phi_{i-1}$ obtained by taking the average of $g$ IID samples $\Delta_{i, 1}, \Delta_{i, 2}, \ldots, \Delta_{i, g} \sim \Delta_i^{\text{single}}$, i.e., $\Phi_i - \Phi_{i-1} \le \frac{1}{g}\sum_{j=1}^g \Delta_{i, j}$.

  We show that $\Delta_i^{\text{single}} \le O(\rho^i)$ with probability $1$. As before, $d_{i-1}$ is some pair-demand $d_{s, t}$ for some endpoints $s, t$ at distance $\ell$ apart. In a (single) randomly sampled LDD, if both endpoints $s, t$ are in the same component, then the demand cancels out (i.e., $\Phi^{\text{single}}_i = 0$), giving us $\Delta_i^{\text{single}} = 0$. On the other hand, if they end up in separate components, we route them to separate component centers (that are at most $\rho^i$ away), making
    \begin{align*}
      \Phi_i^{\text{single}} & = \norm{d_i}_{\opt} + \norm{d_i}_1 \cdot \rho^i \le (\ell + 2 \rho^i) + 2 \cdot \rho^i \\
                             & = (\ell + 2 \rho^{i-1}) + 4\rho^i - 2 \rho^{i-1} \le \Phi_{i-1} + 4 \rho^i .
    \end{align*}
    In other words, $\Delta_i^{\text{single}} \le 4 \rho^i$ with probability $1$, as required. 

    
   Using the Chernoff bound~(\Cref{lem:ChernoffBound}) and applying \Cref{lemma:expectation-analysis}, we have that 
   \[
   \Pr\left[\sum_{j=1}^g \frac{1}{g} \Delta_{i, j} \ge 12 \alpha \cdot \Phi_{i-1}\right] \le \exp( - \frac{1}{3} \cdot 12\alpha \cdot \Phi_{i-1} / Z),
   \]
   where $Z = 4 \rho^i / g$, since $\frac{1}{g} \Delta_{i, j} \sim \frac{1}{g} \Delta_i^{\text{single}} \in [0, \frac{1}{g} 4 \rho^i] = [0, Z]$ with probability $1$. Using this and the fact that $\Phi_{i-1} \ge \rho^{i-1}$ we have that
    \begin{align*}
      \Pr[\Phi_i - \Phi_{i-1} \ge 12 \alpha \cdot \Phi_{i-1}] & \le \Pr\left[\sum_{j=1}^g \frac{1}{g} \Delta_{i, j} \ge 12 \alpha \cdot \Phi_{i-1}\right] \\
                                                              & \le \exp\left(- \frac{1}{3} \cdot 12 \alpha \Phi_{i-1} \frac{g}{ 4 \rho^i }\right) \\
                                                              & \le \exp\left(- 4 \alpha \frac{ \rho^{i-1} \cdot  g}{ \rho^i }\right) \\
                                                              & = \exp\left(- \alpha \frac{g}{ \rho }\right) \\
                                                              & = \exp( - O(\log n) ) = n^{-O(1)} .
    \end{align*}
    Therefore, $\Phi_i \le \Phi_{i-1} + 12 \alpha \cdot \Phi_{i-1} = 13 \alpha \cdot \Phi_{i-1}$ with high probability. 
\end{proof}

Combining the expectation and concentration analysis, we obtain that the potential does not blow up significantly over a single step.
\begin{lemma}[Single-step analysis]\label{lemma:single-step-analysis}
  For any $1 \le i \le \imax$ and any demand $d \in \R^V$ with $\sum_{v \in V} d_v = 0$ the following holds. Fix $d_{i-1} := d$ (with a fixed potential $\Phi_{i-1}$) and run the $i^{th}$ step of \Cref{algo:oblivious-routing-construction}. This induces a new random variable $\Phi_i$. With high probability, $\Phi_i \le 13 \alpha \cdot \Phi_{i-1}$ over the random choices of the $i^{th}$ step.
\end{lemma}
\begin{proof}
  We first note that the routing defined via \Cref{algo:oblivious-routing-construction} is linear, i.e., there exists a linear operator $L = L_i$ ($L$ depends on $i$, but we drop the subscript for ease of notation) such that $d_{i} = L d_{i-1}$.
    
  Using \Cref{lemma:concentration-analysis}, the claim of this lemma holds for all pair-demands $d_{s, t}$ w.h.p. In other words,
  \begin{align*}
    \norm{L d_{s,t}}_\opt + \norm{L d_{s, t}}_1 \cdot \rho^i \le 13\alpha \cdot [ \norm{d_{s,t}}_\opt + \norm{d_{s, t}}_1 \cdot \rho^{i-1} ].
  \end{align*}
  Since there are only $|V|^2$ many pair-demands, it also holds for all demand-pairs w.h.p. We now show that it holds for an arbitrary demand $d_{i-1}$.

  Consider the optimal flow $f \in \R^{\vec{E}}$ that routes $d_{i-1}$ (for an edge $e \in E$, the sign of $f(e)$ denotes the direction of the flow along $e$). Any flow satisfying $d_{i-1}$ can be decomposed into a positive combination of $J$ paths $\{p_j\}_{j=1}^J$ where (1) each path $p_j$ is \emph{unit} in the sense that $p_j(e) \in \{0, 1, -1\}$ for all $e \in E$ (the sign depends on the arbitrary orientation of $e$), and (2) each path $p_j$ starts at a node $s$ where $d_{i-1}(s) > 0$ and ends in a node $t$ where $d_{i-1}(t) < 0$. %
  Therefore, we have that $f = \sum_{j=1}^J \beta_j p_j$ where $\beta_j > 0$ and $p_j \in \R^{\vec{E}}$ is a unit path. Since $p_j$ is a unit path, we have that it is a feasible transshipment solution for some pair-demand $q_j$ (namely, the pair-demand $q_j := d_{s,t}$ where $s$ and $t$ are the endpoints of $p_j$), hence $\norm{q_j}_{\opt} \le W(p_j)$. Due to the flow decomposition, $\norm{d_{i-1}}_{\opt} = W(f) = \sum_{j=1}^J \beta_j W(p_j) \ge \sum_{j=1}^J \beta_j \norm{q_j}_{\opt}$. Furthermore, due to the restriction on the endpoints of the paths we also have that $\norm{d_{i-1}}_1 = \sum_{j=1}^J \beta_j \norm{q_j}_1$.
      
  Using this decomposition into paths, we have that
  \begin{align*}
    \Phi_{i-1} & = \norm{d_{i-1}}_\opt + \norm{d_{i-1}}_1 \cdot \rho^{i-1} \\
               & \ge \sum_{j=1}^J \left( \beta_j \norm{q_j}_{\opt} + \beta_j \rho^{i-1} \norm{q_j}_1 \right) \\
               & \ge \frac{1}{13 \alpha} \cdot \sum_{j=1}^J \left( \beta_j \norm{L q_j}_{\opt} + \beta_j \rho^{i} \norm{L q_j}_1 \right) \\
               & \ge \frac{1}{13 \alpha} \cdot \left[ \norm{L \sum_{j=1}^J \beta_j q_j}_{\opt} + \rho^{i} \norm{L \sum_{j=1}^J \beta_j q_j}_1 \right] \numberthis \label{eq:used-subadd} \\
               & = \frac{1}{13 \alpha} \cdot \left[ \norm{d_i}_{\opt} + \rho^{i} \norm{d_i}_1 \right] = \frac{1}{13 \alpha} \cdot \Phi_i.
  \end{align*}
  Note that in \Cref{eq:used-subadd} we used $\norm{a+b} \le \norm{a} + \norm{b}$ for $\norm{\cdot}_{\opt}$ and $\norm{\cdot}_1$. Rewriting, we have that $\Phi_i \le 13 \alpha \cdot \Phi_{i-1}$ as required, if one assumes the claim for all pair demand, with high probability.
  %
\end{proof}

Finally, with all intermediate steps in place, we prove the main result of this section.
\begin{proof}[Proof of \Cref{theorem:oblivious-routing-approximation}]
  Fix a demand $d_0 := d$ and the value of the optimal solution $\opt := \norm{d}_{\opt}$. Using \Cref{lemma:single-step-analysis}, we have that w.h.p. 
  \begin{align*}
    \Phi_0 &= O(\opt), \\
    \Phi_i &\le (13 \alpha)^i \cdot \opt.
  \end{align*}
  
  In the $i^{th}$ step, for $i \in \{1, \ldots, \imax\}$, let $f_i$ denote the flow which is constructed via \Cref{algo:oblivious-routing-construction}. The total movement $W(f_i)$ can be bounded in the following way:
  \begin{align*}
    W(f_i) \le \rho^i ||d_{i-1}||_1 \le \rho \cdot \Phi_{i-1} \le \rho \cdot (13 \alpha)^{i-1} \cdot O(\opt) \le \exp\left(O(\log n \cdot \log \log n)^{3/4}\right) \cdot \opt .
  \end{align*}

  Summing up over all steps $i \in \{1, \ldots, \imax\}$, the total movement of $f = \sum_{i=1}^{\imax} f_i$ is at most
  \begin{align*}
    W(f) \le \imax \cdot \exp\left(O(\log n \cdot \log \log n)^{3/4}\right) \cdot \opt = \exp\left(O(\log n \cdot \log \log n)^{3/4}\right) \cdot \opt .
  \end{align*}

  We now analyze the final aggregation along an arbitrary spanning tree $T$ towards an arbitrarily chosen root $r$. The total movement of the final aggregation is $n^{C+1} \cdot \norm{d_{\imax}}_1$, where $n^C$ is the largest edge weight (by assumption, $C>0$ is a constant).

  \Cref{lemma:single-step-analysis} implies that $\rho^{\imax} \cdot \norm{d_{\imax}}_1 \le \Phi_{\imax} \le (13 \alpha)^{\imax} \cdot O(\opt)$. Remembering that $\rho^{\imax} \ge n^{C+1}$, we have that the final aggregation contributes $n^{C+1} \cdot \norm{d_{\imax}}_1 \le \exp\left(O(\log n \cdot \log \log n)^{3/4}\right) \cdot \opt$ to the cost. Adding all contributions together, we conclude that the $\ell_1$-oblivious routing yields an $\exp\left(O(\log n \cdot \log \log n)^{3/4}\right)$ approximation.
%
  %
\end{proof}

\subsection{Linear-algebraic interpretation of the routing}\label{sec:linear-algebra-oblivious-routing}

In this section, we develop an algebraic interpretation of the routing matrix $R$ constructed by \Cref{algo:oblivious-routing-construction} via simpler (LDD-induced) routing matrices.

Suppose that $G = (V, E)$ is an undirected graph. Edges are oriented arbitrarily by specifying a linear operator (i.e., matrix) $B\in\R^{V \times \vec{E}}$ which maps a flow $f \in \R^{\vec{E}}$ to the demand $d = B f$ that $f$ routes (defined in \Cref{sec:prelim-transshipment}). On the other hand, an oblivious routing is a linear operator (i.e., matrix) $R \in \R^{\vec{E} \times V}$ that maps a demand $d\in \R^V$ to a flow $f = R d \in \R^{\vec{E}}$.

Given an LDD $\calP$, we define the oblivious routing with respect to $\calP$ as follows. Let the components (partitions) of $\calP$ be $S_1 \subseteq V, \ldots, S_k \subseteq V$ with corresponding centers $\{ c_i \in S_i \}_i$, and let $\{T_i\}_i$ be the shortest path trees of $S_i$ rooted at $c_i$ (with all edges oriented from the root outwards). We define $R(\calP) \in \R^{\vec{E} \times V}$ to be the routing matrix that routes the demand along $T_i$ towards the root. More precisely, $R(\calP) \cdot d$ is the flow that, for each node $s \in S_i$, routes $d(s)$ amount of flow via the leaf-to-root path of $T_i$ (note that the output of $R(\calP) \cdot d$ needs to match the orientation determined by the matrix $B$, hence some components might need to have their sign flipped).

With respect to \Cref{algo:oblivious-routing-construction}, let $R$ denote the oblivious routing performed by the entire procedure, and $R_i$ be the routing that is performed in the $i^{th}$ step, i.e., $R_i$ is performed on $d_{i-1}$ in order to produce the next-step demand $d_i$. We specifically define $R_{\imax+1}$ to be routing in the final aggregation step, i.e., routing along the spanning tree $T$ in Step~\ref{algo-step:spanning-tree} of the algorithm.

By definition of $R_i$, the flow routed in the $i^{th}$ step is $R_i d_{i-1}$. This routes the demand $B R_i d_{i-1}$, hence the residual demand that needs to be routed is $d_i = d_{i-1} - B R_i d_{i-1} = ( I - B R_i ) d_{i-1}$. On the other hand, the routing $R_i$ is constructed by averaging out the routings with respect to $g$ LDDs $\calP^i_1, \ldots, \calP^i_g$. In other words, $R_i = \frac{1}{g} R(\calP^i_1) + \frac{1}{g} R(\calP^i_2) + \ldots + \frac{1}{g} R(\calP^i_g)$.

\begin{lemma}\label{lemma:formula-for-R}
  It holds that
  $$R = \sum_{i=1}^{\imax+1} R_i \cdot (I - B R_{i-1}) \cdot (I - B R_{i-2}) \cdot \ldots \cdot (I - B R_1),$$
  where
  $$R_i = \frac{1}{g} R(\calP^i_1) + \frac{1}{g} R(\calP^i_2) + \ldots + \frac{1}{g} R(\calP^i_g).$$
\end{lemma}
\begin{proof}
  We already argued the second identity.

  Fix some demand $d_0 \gets d$. Let $f$ be the total flow that is routed over all steps when routing $d$, and $f_i$ be the flow that is routed during the $i^{th}$ step of \Cref{algo:oblivious-routing-construction}. Specially, let $f_{\imax}$ be the flow routed in the final aggregation step (Step~\ref{algo-step:spanning-tree}). By definition, we have that $f = \sum_{i=1}^{\imax+1} f_i$. Furthermore, by definition of $R_i$ we have that $f_i = R_i d_{i-1}$. Finally, as argued before, we have that $d_i = ( I - B R_i ) d_{i-1}$. Combining all of these, we have that
  \begin{align*}
    R d & = f = \sum_{i=1}^{\imax+1} f_i = \sum_{i=1}^{\imax+1} R_i d_{i-1} \\
        & = \sum_{i=1}^{\imax+1} R_i (I - B R_{i-1}) d_{i-2} \\
        & = \sum_{i=1}^{\imax+1} R_i (I - B R_{i-1}) \ldots (I - B R_1)d_0 .
  \end{align*}
  In other words, $R = \sum_{i=1}^{\imax+1} R_i (I - B R_{i-1}) \ldots (I - B R_1)$, as required.
\end{proof}

\section{The Distributed Minor-Aggregation Model}\label{sec:dmam}

We contribute to the low-congestion shortcut framework by giving a simple yet powerful interface called the \emph{Distributed Minor-Aggregation model} which makes the recent advances in theoretical distributed computing such as universal optimality~\cite{HWZ21} and oblivious shortcut constructions~\cite{HWZ21,ghaffari2021hop} more accessible. The Distributed Minor-Aggregation Model offers a high-level interface which can be used to succinctly describe many interesting distributed algorithms. In spite of this expressiveness, any algorithm in the Minor-Aggregation model can be efficiently simulated in the standard CONGEST model. To demonstrate this, we describe a simple universally-optimal algorithm for the minimum spanning tree (MST) in \Cref{example:MST-using-dmam}.

\begin{definition}[Distributed Minor-Aggregation Model]\label{def:dmam}
  We are given an undirected graph $G = (V, E)$. Both nodes and edges are individual computational units (i.e., have their own processor and private memory). All computational units wake up at the same time and start communicating in synchronous rounds. Arbitrary local computation is allowed between (each step of) each round. Initially, nodes only know their unique $\tilde{O}(1)$-bit ID and edges know the IDs of their endpoint nodes. Each round of communication consists of the following three steps (in that order).
  \begin{itemize}
  \item \textbf{Contraction step.} Each edge $e$ chooses a value $c_e = \{\bot, \top\}$. This defines a new \emph{minor network} $G' = (V', E')$ constructed as $G' = G / \{ e : c_e = \top \}$, i.e., we contract all edges with $c_e = \top$. Vertices $V'$ of $G'$ are called supernodes, and we identify supernodes with the subset of nodes $V$ it consists of, i.e., if $s \in V'$ then $s \subseteq V$. 

  \item \textbf{Consensus step.} Each individual node $v \in V$ chooses a $\tilde{O}(1)$-bit value $x_v$. For each supernode $s \in V'$ we define $y_s := \bigoplus_{v \in s} x_v$, where $\bigoplus$ is some pre-defined aggregation operator. All nodes $v \in s$ learn $y_s$.

  \item \textbf{Aggregation step.} Each edge $e \in E'$ connecting supernodes $a \in V'$ and $b \in V'$ learns $y_a$ and $y_b$, and chooses two $\tilde{O}(1)$-bit values $z_{e, a}, z_{e, b}$ (i.e., one value for each endpoint). Finally, (every node of) each supernode $s \in V'$ learns the aggregate of its incident edges in $E'$, i.e., $\bigotimes_{e \in \text{incidentEdges(s)}} z_{e, s}$ where $\bigotimes$ is some pre-defined aggregation operator. All nodes $v \in s$ learn the same aggregate value (this might be relevant if the aggregate is non-unique).
  \end{itemize}
\end{definition}

\medskip

\textbf{Distributing the input and output.} Distributed Minor-Aggregation model is simply a communication models upon which one can run various algorithms. The goal is typically to consider a problem like transshipment, SSSP or MST, and design a Minor-Aggregation algorithm that provably terminates with a correct answer in the smallest possible number of rounds. At start, each node/edge receives problem-specific input. This input is distributed among the network in the following way. Such problems are performed on a weighted graph, hence the weight are distributed in a way that each edge $e$ initially knows its weight $w(e)$. Note: the weight does not influence the communication (i.e., always takes $1$ round regardless of the weights). Furthermore, for transshipment, each node $v$ additionally knows its demand value $d(v)$. For SSSP, all nodes and edges additionally know the ID of the source node. Similarly, upon termination, the output is also required to be \emph{stored distributedly}. For example, for the minimum spanning tree (MST) problem, at termination, each edge should know whether it is a part of the MST or not. For SSSP, upon termination, each node should know its distance to the source and each edge whether it is a part of the SSSP tree. For transshipment, upon termination, each edge $e$ should know its part of the flow $f(e)$ and each node $v$ should know its potential $\phi(v)$.

\textbf{Polylogarithmic factors.} We note that the above definition extensively uses the $\tilde{O}$-notation, thereby ignoring polylogarithmic factors (unlike CONGEST which only ignores constant factors). This is due to the fact that the goal of the Minor-Aggregation model is to illuminate the influence of polynomial factors on the runtime of distributed algorithms, while at the same time keeping the framework as simple as possible. Ignoring $\poly(\log n)$ factors is the standard in the literature for distributed global problems like MST or SSSP as the current state-of-the-art is also mostly focused on optimizing polynomial factors and the algorithmic descriptions of global distributed problems get significantly more complicated when one starts optimizing logarithmic factors.


\textbf{Simulation in CONGEST.} While this is not obvious, an algorithm in the Minor-Aggregation model can be efficiently simulated in the standard CONGEST model if one can efficiently solve the part-wise aggregation (PA) problem. Combining this with the very recent work~\cite{GHR21} which solves PA in $\shortcutquality(G) \cdot n^{o(1)}$ rounds (see \Cref{thm:pa-simulation}), any $\tau$-round Minor-Aggregation algorithm can be compiled into a $\tau \cdot \poly(\shortcutquality(G)) \cdot n^{o(1)}$-round CONGEST algorithm. Moreover, on many graph classes this result can be improved down to $\tau \cdot \tilde{O}(\shortcutquality(G))$ (see \Cref{theorem:shortcut-on-specific-graphs}) and potential future improvements in hop-constrained oblivious routing constructions might make lead to unconditional $\tau \cdot \tilde{O}(\shortcutquality(G))$ simulations (which would be near-optimal). The proof of the following simulation theorem is deferred to \Cref{sec:deferred-proofs}. It implies many useful concrete bounds for specific graph classes (see \Cref{theorem:shortcut-on-specific-graphs} for a list).

\begin{restatable}{theorem}{ThmAggregationCongest}\label{thm:dmam}
  Any $\tau$-round Minor-Aggregation algorithm on $G$ can be simulated in the (randomized) CONGEST model on $G$ in $\tau \cdot \poly(\shortcutquality(G)) \cdot n^{o(1)}$ rounds. More generally, if one can solve PA in $\tau_{PA}$ rounds, any $\tau$-round Minor-Aggregation algorithm can be simulated in $\tilde{O}(\tau \cdot \tau_{PA})$ (randomized) CONGEST rounds.
\end{restatable}

\textbf{Operating on minors.} A particularly appealing feature of the Minor-Aggregation model is that any algorithm can be run on a minor of the original communication network in a black-box way. The framework allows the algorithm to be completely unchanged when ran on a minor of a graph rather than on the original graph. Several aspects of the Minor-Aggregation model make operating on minors possible. For instance, in a notable difference with CONGEST, nodes in the Minor-Aggregation model do not know a list of their neighbors\footnote{Moreover, it is not hard to see that a model which allows contractions and gives nodes a list of their neighbors cannot be simulated in CONGEST with $o(n)$ round blowup.}. The following corollary is immediate from the definition.

\begin{corollary}\label{corollary:computation-on-minors}
Let $G=(V,E)$ be an undirected graph, and let $F \subseteq E$ be a subset of edges. Any $\tau$-round Minor-Aggregation algorithm on a minor $G' = G / F$ of $G$ can be simulated via a $\tau$-round Minor-Aggregation algorithm on $G$. Initially, each edge $e \in E$ needs to know whether $e \in F$ or not. Upon termination, each node $v$ in $G$ learns all the information that the $G'$-supernode $v$ was contained in learned.
\end{corollary}

\textbf{An example: MST.} The following example illustrates the expressive power and efficiency of the Minor-Aggregation model framework. We describe an $O(\log n)$-round Minor-Aggregation algorithm for MST. Combining with the simulation \Cref{thm:dmam} and lower bound \Cref{thm:lower-bound-shortcuts}, this implies a $\tilde{O}(1)$-competitive universally-optimal distributed algorithm for MST in CONGEST. We note that, on a weighted planar graph, this algorithm provably completes in $\diameter(G) \cdot n^{o(1)}$ CONGEST rounds; similar results exist for many other graphs, see \Cref{theorem:shortcut-on-specific-graphs}.
\begin{example}[MST]
  \label{example:MST-using-dmam}
  The minimum spanning tree of a weighted graph $G$ can be computed via a $O(\log n)$-round Minor-Aggregation algorithm on $G$.
\end{example}
\begin{proof}
  We can directly implement Boruvka's algorithm~\cite{nevsetvril2001otakar}. In each round, every node $v$ finds the minimum weight $m_v$ of an edge incident to it. This is clearly doable in the the Minor-Aggregation model: no contractions and no consensus are necessary; each edge simply reports its weight to both of its endpoints and the aggregation operator $\bigotimes$ is simply the $\min$ operation.

  We say that an edge $e = \{u, v\}$ is \emph{marked} if its weight is equal to $m_u$ or to $m_v$. Each edge can identify whether it is marked in a single aggregation step (by choosing $y_s := m_s$ for each node $s$). Finally, we add all marked edges to the MST and contract them. We repeat the algorithm on the contracted graph (\Cref{corollary:computation-on-minors}). After $O(\log n)$ iterations, the graph shrinks to a single node and each edge knows whether it is in the MST or not.
\end{proof}

We obtain a $\tilde{O}(\diameter(G))$-round CONGEST algorithm for weighted planar graphs by combining the above $O(\log n)$-round Minor-Aggregation algorithm with the CONGEST simulation \Cref{thm:dmam}, which stipulates that algorithms on planar (or more generally, minor-free) networks can be simulated with a $\tilde{O}(\diameter(G))$ overhead.

\section{Distributed Computation of the $\ell_1$-Oblivious Routing}\label{sec:distributed-implementation-oblivious-routing}

In this section, we show how to distributedly implement \Cref{algo:oblivious-routing-construction} using $n^{o(1)}$ Minor-Aggregation rounds. To this end, we first build a few necessary building blocks in \label{sec:distributed-storage-and-basic-operations} which will simplify the distributed implementation of our algorithms in the remaining sections. We then proceed to prove \Cref{theorem:distributed-evaluation-routing} in \Cref{sec:distributed-R-Rt-evaluation}.

\subsection{Preliminary: distributed storage and basic linear algebra operations}\label{sec:distributed-storage-and-basic-operations}

It is often easier to describe algorithms a purely linear-algebraic setting without going into the low-level details of how and where to store individual values required to specify distributed algorithms. In order to streamline the description of (linear-algebraic) distributed algorithms we first define the notion of distributed storage for the Minor-Aggregation as follows.

\begin{enumerate}
\item We distributedly store a \textbf{node vector} $x \in \R^V$ by storing the value $x_v$ in the node $v \in V$. Similarly, given an \textbf{edge vector} $x \in \R^{\vec{E}}$ we store the value $x_e$ in the edge $e$ (we remind the readers that edges are computational units in the Minor-Aggregation model).

\item We distributedly store a \textbf{spanning subgraph} $H \subseteq G$ (where $G$ is the communication network) by storing distributedly storing the indicator edge vector $x_e = \1{e \in E(H)}$.

\item As explained in \Cref{sec:dmam}, the input to transhipment, namely the weights (edge vector) and the demand $d$ (node vector), are distributedly stored. Each node $v$ knows its part of the demand $d(v)$. Upon output, we require the flow $f$ (edge vector) and potential $\phi$ (node vector) to be distributedly stored. The specification is analogous for SSSP: on input, we distributedly store the edge weights and require that all nodes know the ID of the source $s$. Upon output, we require the SSSP tree (spanning subgraph) and distances (node vector) to be distributedly stored.

\end{enumerate}

Furthermore, we verify that the following linear-algebraic graph operations can be evaluated quickly. In the following suppose $a, b$ are two distributedly stored node or edge vectors.
\begin{enumerate}
\item \textbf{Vector addition and component-wise transforms.} Without any extra communication we can compute and distributedly store (1) the sum vector $a + b$, or (2) the vector $( f(a_i) )_i$, i.e., the vector $a$ with $f : \R \to \R$ applied component-wise. 
  
\item \textbf{Dot products and $\ell_p$ norms.} With a single Minor-Aggregation round we can compute and broadcast to all nodes the value of (1) the dot product $a^T b$, or (2) $\ell_p$-norm of a vector, i.e., $(\sum_i |a_i|^p)^{1/p}$. We show why for the case of dot products. Suppose $a, b$ are node vectors. Each node $v$ evaluates $a_v \cdot b_v$ and stores it as its private input $x_v$. We contract all edges of the graph and apply the consensus step with the plus-operator, which broadcasts the dot product $\sum_{v \in V} a_v \cdot b_v = a^T b$ to all nodes. For the case of edge vectors, we contract all edges, and perform an aggregation step with the $+$-operator where each edge $e = \{a, b\}$ computes $a_e \cdot b_e$ and sets $z_{e, a} = a_e \cdot b_e$, $z_{e, b} = 0$. This informs all nodes about $a^T b$. Computing $\ell_p$ norms is analogous.

\item \textbf{Multiplication with $B$ and $B^T$.} With a single Minor-Aggregation round we can compute and distributedly store $B \cdot a$ and $B^T \cdot b$, where $a \in \R^{\vec{E}}$ is an edge vector and $b \in \R^V$ is a node vector. To compute $B \cdot a$, we leave all edges uncontracted and perform an aggregation step with the $+$-operator, where each edge $e = \{u, v\}$ with orientation $\vec{e} = (u, v)$ sets $z_{e,u} = +a_u$ and $z_{e, v} = -a_v$. Upon completion of the step, a node $w \in V$ learns $(B \cdot a)_w$, as required for distributed storage. To compute $B^T \cdot b$, we perform a consensus step without any contracted edges with each node $v$ setting $y_v = x_v = b_v$ as its private input. Each edge $\vec{e} = (u, v)$ learns $b_u$ and $b_v$ and computes $(B^T b)_e = b_u - b_v$, as required.

\item \textbf{Multiplication with $W$.} With a single Minor-Aggregation round we can compute and distributedly store $W \cdot a$, where $a \in \R^{\vec{E}}$ is an edge vector. 
Since for each edge $e = (u, v)$, $a_e$ and $w_e$ are stored at both vertex $u$ and $v$,
every node $u \in V$ learns $W\cdot a$ for all the edges incident to $u$.
\end{enumerate}

\subsection{Distributed evaluation of $R$ and $R^T$}\label{sec:distributed-R-Rt-evaluation}

In this section we show the following result.

\theoremDistributedEvaluationRouting*

The rest of the section is dedicated to proving \Cref{theorem:distributed-evaluation-routing}. Inspecting \Cref{algo:oblivious-routing-construction}, we require distributed implementations of a few subroutines. First, we require a way to efficiently sample from the LDD distribution. To this end, we leverage the following theorem from prior work~\cite{HL18}.
\begin{restatable}{lemma}{lemmaDistributedLddSampling}[LDD sampling~\cite{HL18}] 
  \label{theorem:distributed-LDD-sampling}
  Suppose $G$ is a weighted graph $G$ with weights in $[1, n^{O(1)}]$. For any $\rho \ge 1$, there exists a distributed algorithm which samples an LDD from a distribution of radius $\rho$ and quality $\exp(O(\sqrt{\log n \cdot \log \log n}))$ in $\exp(O(\sqrt{\log n \cdot \log \log n}))$ rounds of Minor-Aggregation.

  Upon termination, each node $v$ knows the center node $c_i$ of its LDD component $S_i$, and $v$'s parent edge in the shortest path tree of $G[S_i]$ that is centered at $c_i$.
\end{restatable}
\begin{proof}Deferred to \Cref{sec:deferred-proofs}.\end{proof}

In order to compute the flows obtained by routing the demand $d_{i-1}$ using the shortest path trees in each LDD component, we need to be able to quickly compute subtree sums of rooted trees. This is furnished by the following result from \cite[Theorem 5.1, full version]{dory2019improved}.

\begin{lemma}[Subtree sum~\cite{dory2019improved}]
  \label{theorem:subtree-sum}
  Let $G = (V, E)$ be a graph and let $F \subseteq G$ be a collection of node-disjoint rooted trees (i.e., a rooted forest) that are subgraphs of $G$. Initially, $F$ is stored distributedly (each edge know whether $e \in F$ and its orientation). Furthermore, each node $v$ has a $O(\log n)$-bit private input $x_v$. There exists a $\tilde{O}(1)$-round Minor-Aggregation algorithm such that, upon termination, each node $v$ learns both the sum of private values of all of its descendents and the sum of private values of all of its ancestors.
\end{lemma}

\begin{lemma}\label{lemma:evaluating-LDD-routing-left-and-right}
  Suppose an LDD $\calP$ over a weighted graph $G$ is distributedly stored (in the sense of \Cref{theorem:distributed-LDD-sampling}). There exists a $\tilde{O}(1)$-round Minor-Aggregation algorithm that, given distributedly stored $d \in \R^V$ and $c \in \R^{\vec{E}}$, evaluates and distributedly stores $R(\calP) \cdot d \in \R^{\vec{E}}$ and $R(\calP)^T \cdot c \in \R^V$.
\end{lemma}
\begin{proof}
  We consider the shortest path trees $T_1, \ldots, T_k$ for each component of the LDD. To compute $R(\calP)\cdot d$, we need to output $(R(\calP) \cdot d)_e = 0$ for all non-tree edges. For each tree-edge edge $e = \{v, u\} \in E(T_i)$ where $u$ is closer to the root of $T_i$, we output $(R(\calP) \cdot d)_e$ to be the total sum of demands $d(w)$ over all $w$ that are descendents of $v$. This number needs to possibly be negated if $\vec{E}$ contains the edge in the opposite orientation (i.e., if $(u, v) \in \vec{E}$). The correctness of this procedure is immediate. We can directly implement this in $\tilde{O}(1)$ rounds of Minor-Aggregation via \Cref{theorem:subtree-sum}.

  Similarly, to compute $R(\calP)^T \cdot c$, each node $v \in V(T_i)$ needs to compute the sum of $c(e)$ over all edges $e$ (with some elements possibly negated) on the root-to-$v$ path (in $T_i$). Again, we can directly implement this in $\tilde{O}(1)$ rounds of Minor-Aggregation via \Cref{theorem:subtree-sum}.
\end{proof}

Finally, we are ready to prove the main theorem of this section.

\begin{proof}[Proof of \Cref{theorem:distributed-evaluation-routing}]
  The $\exp(O(\log n \cdot \log \log n)^{3/4})$-approximation guarantee of the $\ell_1$-oblivious routing constructed by \Cref{algo:oblivious-routing-construction} is ensured by \Cref{theorem:oblivious-routing-approximation}.

  We examine \Cref{algo:oblivious-routing-construction} and show we can implement each ingredient required for the procedure. Firstly, we sample and distributedly store all LDDs $\calP^i_j$ for $1 \le i \le \imax$, $1 \le j \le g$. This is provided by \Cref{theorem:distributed-LDD-sampling} in $g \cdot \imax \cdot \exp(\sqrt{\log n \cdot \log \log n}) = \exp(O(\log n\cdot \log \log n)^{3/4})$ rounds of Minor-Aggregation. Furthermore, we compute and store an arbitrary spanning tree $T$ of $G$. A simple choice is to use the MST, whose construction can be computed in $\tilde{O}(1)$ rounds of Minor-Aggregation, as stipulated by \Cref{example:MST-using-dmam}.

  Let $R_i = \frac{1}{g} R(\calP^i_1) + \frac{1}{g} R(\calP^i_2) + \ldots + \frac{1}{g} R(\calP^i_g)$ (\Cref{lemma:formula-for-R}). We note we can evaluate (and distributedly store) $R_i \cdot d$ (which we call \emph{multiplication from right}) and $R_i^T \cdot c$ (called \emph{multiplication from left}) in $\exp(O(\log n\cdot \log \log n)^{3/4})$ rounds for any distributedly stored vectors $c, d$. We simply evaluate each term in the sum (when multiplied from both left and right), and add them together. This contributes $\tilde{O}(1) \cdot g$ rounds of communication by utilizing \Cref{lemma:evaluating-LDD-routing-left-and-right} and basic operations of \Cref{sec:distributed-storage-and-basic-operations}. Furthermore, we can evaluate and distributedly store $(I - B R_i) \cdot d$ and $(I - B R_i) \cdot d$ for any vector $d \in \R^V$ in $\tilde{O}(g)$ rounds using the basic operations of \Cref{sec:distributed-storage-and-basic-operations}. Specially, we define $R_{\imax+1}$ to be the routing with respect to a spanning tree $T$. Since routing along $T$ can be seen as routing along a simple LDD, multiplication of $R_{\imax+1}$ from left and right in $\tilde{O}(1)$ rounds is furnished by \Cref{lemma:evaluating-LDD-routing-left-and-right}.

  Finally, we remember from \Cref{lemma:formula-for-R} that $R = \sum_{i=1}^{\imax+1} R_i \cdot (I - B R_{i-1}) \cdot (I - B R_{i-2}) \cdot \ldots \cdot (I - B R_1)$. We already showed that we can multiply each factor from both left and right in $\tilde{O}(g)$ steps, hence we can compute $R d$ and $R^T c$ in $\imax^2 \cdot \tilde{O}(g) = \exp(O(\log n\cdot \log \log n)^{3/4})$, as required.
\end{proof}

\section{Distributed $(1+\eps)$-Transshipment}
In this section, we present our distributed algorithm for approximating the transshipment problem. The result is shown in Theorem \ref{theorem:distributed-transshipment}, and for the sake of completeness, we restate it below. 
\theoremDistributedTransshipment*

We follow Sherman's framework~\cite{She17b} via the multiplicative weights update~(MWU) paradigm \cite{AroraHK12} based on the construction of efficient $\ell_1$-oblivious routing presented in the last section, and then give the round complexity of the implementation in the CONGEST model. Instead of exactly performing Sherman's method~\cite{She17b}, our algorithm does not apply Bourgain's $\ell_1$-embedding~\cite{Bou85} when constructing the $\ell_1$-oblivious routing operator.

Given an $\ell_1$-oblivious routing, we solve transshipment by utilizing its \emph{boosting} property: the $\ell_1$-oblivious routing yields an $n^{o(1)}$-approximate solution for transshipment, which we can then \emph{boosting} to an $(1+\eps)$-approximate one using multiplicative weights (or gradient descent). This boosting property was implicitly used in many papers~\cite{She17b,BKKL17,Li20,AndoniSZ20}, and was explicitly isolated in a recent writeup~\cite{zuzic2021simple} that shows any black-box dual approximate solution can be boosted to an $(1+\eps)$-approximate one. The following lemma gives an end-to-end interface to boosting: given an $\ell_1$-oblivious routing, we can construct $(1+\eps)$-approximate solutions to transshipment.

\begin{lemma} \label{lem:mainTransshipment}
  Let $R$ be an $\alpha$-approximate $\ell_1$-oblivious routing with respect to a transshipment instance on a weighted graph $G$. Suppose we can compute matrix-vector products with $R$ and $R^T$ in $M$ Minor-Aggregation rounds. Then we can compute and distributedly store a flow $f \in \R^{\vec{E}}$ and a \emph{vector of potentials} $\phi$ in $\tilde{O}(\alpha^2\eps^{-2}M)$ rounds such that
  \begin{enumerate}
  \item $\norm{Wf}_1 \leq (1+\eps)d^{\top} \phi \leq (1+\eps) \norm{d}_{\opt}$;
  \item $\norm{Bf - d}_{\opt} \leq \eps \norm{d}_{\opt}$;
  \item $\norm{W^{-1} B^T \phi}_\infty \le 1$.
  \end{enumerate}
\end{lemma}

Before proving Lemma \ref{lem:mainTransshipment}, we give the following lemma which returns the rough flow $f$ and potential $\phi$ affected by the input parameter $t$.
The algorithm we simulate is summarized in Algorithm \ref{algo:transshipment-mwu}.
In Algorithm \ref{algo:transshipment-mwu}, MWU is used to determine if one region is approximately feasible. It returns a flow $f$ or a potential  $\phi$ that satisfy the two conditions in the following lemma. 

\begin{lemma}[Lemma C.2 in \cite{Li20}] \label{lem:TransshipmentMWU}
  Consider a transshipment instance with demand vector $d$ and a parameter $t \geq \norm{d}_{\opt}/2$. Let $R$ be an $\alpha$-approximate $\ell_1$-oblivious routing operator. 
  Suppose we can compute matrix-vector products with $R$ and $R^T$ in $M$ Minor-Aggregation rounds. 
  Then there is an $\tilde O(\alpha^2\eps^{-2}M)$-round Minor-Aggregation algorithm outputs either
  \begin{enumerate}
  \item an acyclic flow $f$ satisfying $\norm{Wf}_{1} \leq t$ and $\norm{WRBf - WRd}_1 < \eps t$, or
  \item a potential $\phi$ with $d^{\top} \phi = t$.
  \end{enumerate}
\end{lemma}
\begin{proof}

  We can show an algorithm for the above lemma applying the MWU method~\cite{AroraHK12} for (approximately) checking the feasibility of the following region
  \begin{equation}
    \left\{ y \in \R^{\vec{E}} \bigg| \left\| y^{\top} W R B W^{-1} \right\|_{\infty} + \frac{1}{t} y^{\top} W R d \leq - \eps \text{ and } \norm{y}_{\infty} \leq 1 \right\}.
  \end{equation}
  Observe that this problem is tightly connected to the dual convex-programming formulation of the transshipment problem. Building upon the presentation of Li~\cite{Li20}, we give a self-contained procedure summarized in Algorithm~\ref{algo:transshipment-mwu} for checking the feasibility of the region. The correctness of \Cref{algo:transshipment-mwu} is proved in \cite{Li20}. Furthermore, one can easily check that all the operations required to implement \Cref{algo:transshipment-mwu} in the Minor-Aggregation model are either matrix-vector multiplications or basic operations presented in \Cref{sec:distributed-storage-and-basic-operations}.
\end{proof}

\begin{algorithm}[H]
    \caption{MWU for transshipment using $\ell_1$-oblivious routing}
    \label{algo:transshipment-mwu}

    \begin{enumerate}
      \itemsep0em 
    \item Set $\delta \gets \eps/(2\alpha)$.
    \item Set $p^{+}_0(e) \gets 1/(2m)$ and $p^{-}_0(e) = 1/(2m)$ for all $e \in E$.
    \item Set $\phi^{+}_0(e) = 1$ and $\phi^{-}_0(e) = 1$ for all $e \in E$.
    \item Define $\chi_e := (\mathbb{1}_u - \mathbb{1}_v)$, where $e=(u,v) \in E$, and $\mathbb{1}_u,\mathbb{1}_v$ are indicator vectors.
    \item For $t=1,2,\ldots, T$ where $T=O(\alpha^2\eps^{-2}\log{m})$: 
      \begin{enumerate}
      \item If $\norm{WRB \sum_{e \in E} w^{-1}_e (p^+_{t-1}(e) \chi_e - p^-_{t-1}(e) \chi_e) + \frac{1}{t}WRd}_1 \geq \eps$
        \begin{enumerate}
        \item Set $y(t) \gets -\mathrm{sign} \left( WRB \sum_{e \in E} w^{-1}_e (p^+_{t-1}(e) \chi_e - p^-_{t-1}(e) \chi_e) + \frac{1}{t}WRd \right)$
        \end{enumerate}
      \item Otherwise, set $f \gets -t \sum_e w_e^{-1} (p^+_e \chi_e - p^-_{e} \chi_e)$ and output $f$.
      \item For each $e \in E$:
        \begin{enumerate}
        \item Set $\phi^{+}_t(e) \gets \frac{1}{2} \phi^{+}_{t-1}(e) \cdot \exp{(\delta \cdot (y_{t}^{\top} WRB w_e^{-1} \chi_e + \frac{1}{t}y_{t}^{\top}WRd))}$.
        \item Set $\phi^{-}_t(e) \gets \frac{1}{2} \phi^{-}_{t-1}(e) \cdot \exp{(\delta \cdot (-y_{t}^{\top}WRB w_e^{-1} \chi_e + \frac{1}{t}y_{t}^{\top}WRd))}$.
        \end{enumerate}
      \item For each $e \in E$, set $p^{+}_{t}(e) \gets \phi^{+}_t(e)/\sum_{e \in E} \phi^{+}_t(e)$, $p^{-}_{t}(e) \gets \phi^{-}_t(e)/\sum_{e \in E} \phi^{-}_t(e)$.
      \end{enumerate}
    \item Set $x \gets \frac{1}{T}\cdot \sum_{i \in [T]} y_{t}$.
    \item Set $\phi$ to be the vector $-(x^{\top}WR)^{\top}$ scaled up so that $\phi^{\top}d = t$.
    \item Output $\phi$.
    \end{enumerate}
  \end{algorithm}





Now we prove Lemma \ref{lem:mainTransshipment} based on Lemma \ref{lem:TransshipmentMWU} by doing a careful binary search on the value of $t$. 
\begin{proof}[Proof of Lemma~\ref{lem:mainTransshipment}]
We first describe our algorithm. Begin with $t = \textrm{poly}(n)$, which is an upper bound on $\norm{d}_{\opt}$. As long as Algorithm~\ref{algo:transshipment-mwu} with parameter $t$ returns a flow $f$, we decrease the value of $t$ by setting $t \gets \frac{1+\eps}{2} t$ and invoke Algorithm~\ref{algo:transshipment-mwu} with this new value of $t$. At some point Algorithm~\ref{algo:transshipment-mwu} must return a potential $\phi$ for which $\norm{d}_{\opt} \geq d^{\top} \phi = t$. This means that $\norm{d}_{\opt} \in (t,2t)$ and we can run binary search in this interval to compute two values $t_\ell, t_r$ such that (1) $ t_r - t_\ell \leq \eps \norm{d}_{\opt} $ and (2) $ \norm{d}_{\opt} \in (t_\ell, t_r).$ Finally, we run Algorithm~\ref{algo:transshipment-mwu} with parameter $t = t_\ell/(1+\eps)$ to obtain potentials $\phi$ and then run Algorithm~\ref{algo:transshipment-mwu} with parameter $t = t_r$ to obtain a flow $f$. The algorithm returns the flow-potential pair $(f, \phi)$. The property (3) in Lemma~\ref{lem:mainTransshipment} can be found in Claim C.3 of \cite{Li20}.
  
We next give the round complexity. 
By 
Lemma \ref{lem:TransshipmentMWU}, for each demand vector $d$ and a parameter $t$, Algorithm~\ref{algo:transshipment-mwu} takes $\tilde O(\alpha^2\eps^{-2}M)$ Minor-Aggregation rounds. By the above analysis, the binary search of $t$ takes $O(\log{(n/\eps)})$ times. Since each matrix-vector multiplication takes $M$ Minor-Aggregation rounds, we conclude that the round complexity of Lemma \ref{lem:mainTransshipment} is $\widetilde{O}(\alpha^{2}\eps^{-2}M)$. 
\end{proof}

Finally, combining Sherman's transshipment framework with our distributed graph-based $\ell_1$-oblivious routing, we can approximate transshipment.
To make sure that the demand is satisfied exactly, we select a carefully-chosen approximation factor within the framework, followed by routing the residual demand via the $\ell_1$-oblivious routing.

\begin{proof}[Proof of \Cref{theorem:distributed-transshipment}]
Our algorithm first computes an $\ell_1$-oblivious routing operator $R$ with approximation factor $\alpha = \exp(O(\log{n}\log\log{n})^{3/4})$, and then computes an $(1+\eps / (2\alpha))$-approximate transshipment solution $f$ and $\phi$ based on $R$. 
To make sure that demand is satisfied, we use the $\ell_1$-oblivious routing operator $R$ to route the residual demand $d - Bf$ to obtain $f'$. 

We show that flow vector $f+ f'$ and potential vector $\phi$ satisfying the requirement. 
Since $Bf' = d-Bf$, we have $B(f + f') = Bf + d - Bf= d$. 
In addition, we have 
\begin{align*}
\norm{W(f + f')}_1\leq & \norm{Wf}_1 + \norm{Wf'}_1 \\
\leq & \left(1+\frac{\eps }{2\alpha}\right) \norm{d}_{\opt} + \alpha \|d - Bf\|_\opt \\
\leq & \left(1 + \frac{\eps }{ 2\alpha}\right)  \norm{d}_{\opt} + \alpha \cdot \frac{\eps}{2\alpha} \norm{d}_\opt\\
< & (1+\eps) \|d\|_\opt.\end{align*}

The round complexity is obtained by our parameter setting, Theorem~\ref{theorem:distributed-evaluation-routing}, and Lemma~\ref{lem:mainTransshipment}.
\end{proof}

\section{Distributed $(1+\eps)$-SSSP}
In this section, we present our distributed algorithm to construct a single source shortest path (SSSP) tree and prove the following result. 
\thmSSSP*
Our algorithm is obtained by a distributed implementation of a simplified version of the SSSP algorithm presented in~\cite{Li20} and mainly included in this paper for completeness.

We first review the SSSP algorithm of \cite{Li20} in \Cref{sec:sssp-algo-sequential}, and then give its distributed implementation in \Cref{sec:distributed-sssp}. 

\subsection{SSSP Algorithm}\label{sec:sssp-algo-sequential}

In ~\cite{Li20}, Li presented an algorithm to construct an approximate SSSP tree given an algorithm that approximate the transshipment problem. We start by defining the notion of \emph{expected} single source shortest path (ESSSP) tree. 

\begin{definition}[Definition D.1 in \cite{Li20}]
  Given a graph $G=(V, E)$, a source $s \in V$ and a demand vector $d \in \R^V$ with $d_v \ge 0$ for each $v \neq s$, an $\alpha$-approximate \textbf{expected SSSP (ESSSP)} tree is a tree $T$ such that
  \[
    \mathbb{E}\left[\sum_{v : d_v>0}d_v\cdot \dist_T(s,v)\right]\le\alpha\sum_{v : d_v>0}d_v\cdot \dist_G(s,v),
  \]
  where $\dist_T(s,v)$ and $\dist_G(s,v)$ denote the distances between $s$ and $v$ in tree $T$ and graph $G$ respectively. 
\end{definition}

In \cite{Li20}, Li gave an algorithm to compute an ESSSP tree with respect to a given graph, a source, and a demand vector. 
The algorithm is summarized in Algorithm~\ref{algo:ESSSP}.
On a high level, the ESSSP algorithm first  obtains a flow vector, that is an approximate transshipment solution for the given graph and demand vector, and then samples an outgoing edge (with respect to the flow vector) for each vertex with probability proportional to the outgoing flow value. 
If the flow vector is acyclic, then the sampled edges form a ESSSP tree. 
But if the flow vector contains some directed cycles, then the sampled edges form some connected components such that each connected component contains exactly one directed cycle. For this case, the algorithm iteratively contracts connected components of the sampled graph in the input graph, and recurses on the new graph until the result is a directed tree. In the end, the algorithm uses sampled edges in all the recursions to construct an ESSSP tree.   
The correctness of Algorithm~\ref{algo:ESSSP} is proved in Claim D.8 of
\cite{Li20}. 

\begin{lemma}[Correctness of ESSSP, Claim D.8, \cite{Li20}]\label{lem:ESSSP}
  Given a graph $G=(V, E)$, a source $s \in V$, a demand vector $d \in \R^V$ with $d_v\ge0$ for each $v\neq s$, and a parameter $\eps$,  
  Algorithm~\ref{algo:ESSSP} computes an $(1+\eps)$-approximate ESSSP tree $T$. 
  Furthermore, the recursion depth of the algorithm is $O(\log n)$. 
\end{lemma}

\begin{algorithm}[h]
  \caption{ESSSP($G=(V, E, w)$, $s$, $d$, $\eps$) \ \ \ (Algorithm 4 of \cite{Li20}, restated)}
  \label{algo:ESSSP}
  \begin{enumerate}
      \item  Set $\eps' = c \eps / \log n$ for some constant $c$, and compute the flow $f$ and the potential $\phi$ of an $(1+\eps')$-approximate transshipment on $G$ with demand vector $d$.
  \item Each vertex $u\in V\setminus\{s\}$ samples an edge $(u,v)$ such that $f(u,v)>0$ with probability $f(u,v)/\sum_{f(u,v)>0}f(u,v)$. Let $H$ be the directed graph consisting of the sampled edges and directed self-loop $(s, s)$.
      \item For each connected component $C$ of $H$ (ignoring edge directions when computing connected components): 
      \begin{enumerate}
          \item Let $w(C)$ be the total weights of the edges in the (unique) cycle in $C$.
          \item Let $T_C$ be the graph formed by contracting the (unique) cycle in $C$ into a supervertex $x_C$.
          \item Let $v_C$ be a supervertex for $C$ with demand $d'_{v_C}\assign\sum_{v\in V(C)}d_v$.
      \end{enumerate}
      \item For each edge $(u,v)\in E$: 
      \begin{enumerate}
          \item Let $C$ and $C'$ be the connected components of $H$ containing $u$ and $v$ respectively. 
          \item If $C\neq C'$:
          \begin{enumerate}
              \item $e'\assign(u,v)$.
              \item $w'(u,v)\assign w(u,v)+\dist_{T_C}(u,x_C)+\dist_{T_{C'}}(v,x_{C'})+w(C)+w(C')$.
          \end{enumerate}
      \end{enumerate}
      \item Let $s'\assign v_{C_s}$, where $C_s$ is the component in $H$ containing $s$.
      \item Denote graph $G'=(V', E', w')$ with $V'=\{v_C\}$ and $E'=\{e'\}$, where $C$ is the component in $H$.
      \item $T' \assign\text{ESSSP}(G', s', d', (1 + \eps) / (1 + 3\eps') - 1)$.
      \item Initialize $T\assign\emptyset$. 
      \item For each edge $(u,v)\in T'$:
      \begin{enumerate}
          \item $T\leftarrow T\cup(u,v)$.
      \end{enumerate}
      \item For each connected component $C$ of $H$:
      \begin{enumerate}
          \item Remove an arbitrary edge from the (unique) cycle in $C$ and merge the resulting tree with $T$. 
      \end{enumerate}
    \item Output $T$. 
  \end{enumerate}
\end{algorithm}

\textbf{Converting ESSSP to an SSSP tree.} The SSSP construction algorithm is summarized in Algorithm~\ref{algo:TStoSSSP} (SSSP), which is simplified version of Algorithm 6 from~\cite{Li20}. 
On a high level, the algorithm tries to finds an $(1 + \eps)$-ESSSP tree $T^*$, which can be shown to provide a $(1 + \eps)$-approximate path to at least half of the nodes (appropriately weighted). We can identify which half of the nodes the path induces via the transshipment potentials $\phi$. When a good path to a node $v$ is found, we can remove $v$ from the \emph{target set} of nodes $V'$ (which initially starts with $V' \gets V$) and start a new iteration with a smaller $V'$.

This procedure intuitively produces $O(\log n)$ trees such that for each node $v$ at least one tree offers an $(1+\eps)$-optimal $s$-to-$v$ path. There is a simple ``trick'' to obtain a single tree that is good with respect to all nodes (described in \cite{HL18}, Algorithm ExpectedSPDistance and proven in Lemma 14): each node keeps track of its parent pointer $\parent(v)$. If a new tree $T^*$ offers a better path to $v$ than previously known, we reassign $\parent(v)$ to point to the parent of $v$ in $T^*$.




Algorithm~\ref{algo:TStoSSSP} (SSSP) is a simplified version \cite[Algorithm 6]{Li20}, due to the fact that we do not require an $\ell_1$-embedding to solve a transshipment instance.  
The correctness of Lemma~\ref{lem:TStoSSSP} follows from Claim E.1 and Claim E.2 in~\cite{Li20}. 
\begin{lemma}[Correctness of SSSP, Claim E.1 and Claim E.2, \cite{Li20}]\label{lem:TStoSSSP}
  Given a graph $G=(V, E)$, a source $s \in V$, and a parameter $\eps$,  
  Algorithm~\ref{algo:TStoSSSP} computes an $(1+\eps)$-approximate SSSP tree $T$. 
  Furthermore, the number of iterations of the while loop is $O(\log n)$ with high probability.
\end{lemma}

\begin{algorithm}[h]
  \caption{SSSP($G=(V, E, w)$, $s \in V$, $\eps > 0$)}
  \label{algo:TStoSSSP}
  \begin{enumerate}
      \item  Initialize $V' \gets V\setminus\{s\}$, potential vector $\phi\assign \vec{0}$, and \emph{parent pointers} $p : V \to V$ initially $p(v) \gets v$. 
      \item While $V'\neq\emptyset$:
      \begin{enumerate}
          \item Let $d\assign\sum_{v\in V'}(\mathbb{1}_v-\mathbb{1}_s)$.
          \item Obtain a $(1 + \frac{\eps}{10})$-apx flow-potential pair $(f^*, \phi^*)$ for transshipment. 
          \item Obtain $T^* \assign\text{ESSSP}(G, s, d, \Theta(\eps/\log{n}))$. 
          \item Root $T^*$ at $s$ and compute distances $\dist_{T^*} (s,v)$ for all $v$.
          \item Let $T = (V, \{ (v, \parent(v)) \}_v )$ be the tree defined by parent pointers $\parent$. 
          \item Compute the distances $\dist_T(s, v)$ for all $v$ (with weights $w$). 
          \item For each vertex $v\in V$:
          \begin{enumerate}
              \item If $\dist_{T^*}(s,v) \le \dist_T(s, v)$:
              \begin{enumerate}
                  \item set $\parent(v) \gets $ parent of $v$ in $T^*$.
              \end{enumerate}
              \item $\phi(v)\assign\max\{\phi^*(v), \phi(v)\}$ where $\phi^*$ is translated so that $\phi^*(s) = 0$. 
              \item If $\dist_T(s,v)\le(1+\eps)\phi(v)$:
              \begin{enumerate}
                  \item $V'\assign V'\setminus\{v\}$. 
              \end{enumerate}
          \end{enumerate}
      \end{enumerate}
     \item Output tree $T = (V, \{ (v, \parent(v)) \}_v )$ defined by parent pointers $\parent$.
  \end{enumerate}
\end{algorithm}





\subsection{Distributed implementation of SSSP}\label{sec:distributed-sssp}
We give our distributed implementation of Algorithm~\ref{algo:TStoSSSP} (SSSP), and prove Theorem~\ref{thm:SSSP}.
We start with some useful subroutines that follow from prior work~\cite{dory2019improved}. 
The following lemma roots (i.e., orients the edges) given an unrooted forest and roots for each connected component. 
\begin{lemma}[Rooting a tree, \cite{dory2019improved}]\label{lem:sssp_basic_operations}
  Let $G$ be a graph, and $G' \subseteq G$ be a distributedly stored undirected forest and suppose that in each connected component all nodes agree on a (so-called) root node. There is a $\tilde{O}(1)$-round Minor-Aggregation algorithm to root the each tree in $G'$ at its root (i.e., each edge computes its direction towards its component's root), and to compute the distances between any vertex and its root in $G'$.
\end{lemma}
\begin{proof}
  Note that we can assume without loss of generality that $G'$ is a tree (i.e., there is a single connected component) since we can independently run algorithms on different connected components. Connected components can be easily identified via a single round by contracting all edges and using $\min$-aggregation.
  
  To root a single tree we utilize Theorem 5.3 in the full version of \cite{dory2019improved} (which can be reinterpreted as constructing a heavy-light decomposition in $\tilde{O}(1)$ rounds of Minor-Aggregation). We only use the fact that, upon exit, each edge learns the orientations towards the root. Computing the distance to the root is performed via \Cref{theorem:subtree-sum}.
\end{proof}

We give a distributed algorithm to find all the directed cycles in a directed minor of the communication network if every vertex in the minor has one outgoing edge.
The algorithm is summarized in Algorithm~\ref{algo:cycle}. 

The algorithm iteratively samples vertices with constant probability, and contracts other vertices to the sampled vertices until a cycle can be easily identified. Then the algorithm backtracks the identified cycle through the contraction process. 

\begin{algorithm}
  \caption{FindCycles$(G' = (V, E))$ where $G'$ is a directed graph where each node has one outgoing edge}
  \label{algo:cycle}
  
  \begin{enumerate}
  \item Find all the connected components of $G'$ (ignoring edge directions only in this step).
  \item For each connected component $C$ of $G'$, if $C$ contains a self-loop, then return the self-loop, otherwise
    \begin{enumerate}
    \item Let $i$ be $0$ and  $C_0$ be $C$.  Every vertex in $C_i$ samples itself with probability $1/10$. 
    \item Repeat until there is a sampled vertex $v$ in $C_i$ such that the walk from $v$ along outgoing edges finds a cycle before visits another sampled vertex.
      \begin{enumerate}
      \item For each directed edge $e = (u, v)$ of $C_i$ such that $u$ is sampled and $v$ is not sampled, 
        set $c_e = \bot$. Set $c_e = \top$ for all the other edges. 
      \item Set $C_{i+1} = C_{i} / \{e: c_e = \top\}$, and $i = i + 1$.
      \item Every vertex in $C_i$ samples itself with probability $1/10$. 
      \end{enumerate}
    \item Let $v$ be an arbitrary vertex whose walk finds a cycle along outgoing edges, and $E_i$ be the edges of the cycle obtained from the walk from $v$.
    \item For $j = i - 1$ down to $0$ 
      \begin{enumerate}
      \item Let $E_{j}$ be $E_{j+1}$'s  corresponding edges in $C_{j}$, and $V_j$ be the endpoints of $E_j$ in $C_j$. 
      \item For each vertex of $C_{j+1}$ which corresponds to an induced subgraph of $C_{j}$ that contains   two vertices $u, v$ of $V_j$, add the edges on the directed path between $u$ and $v$ in $C_j$ to $E_j$.
      \end{enumerate}
    \item Return $E_0$.   
    \end{enumerate}
  \end{enumerate}
\end{algorithm}

\begin{lemma}[Finding cycles]\label{lem:cycle}
  Let $G$ be a graph, and $G' \subseteq G$ be a distributedly stored directed subgraph (each edge knows whether $e \in E(G')$ and its direction) such that every vertex of $G'$ has one outgoing edge (the outgoing edge can be a self-loop). 
  There is an $\tilde{O}(1)$ round Minor-Aggregation algorithm to find all the directed cycles in $G'$ (i.e., every vertex learns the incident edges belonging to a cycle of $G'$). 
\end{lemma}
\begin{proof}
  We first show that Algorithm~\ref{algo:cycle} finds all the directed cycles in $G'$.
  Since every vertex has one outgoing edge, every connected component of $G'$ identified in Step 1 contains exactly one cycle. 
  For each connected component $C$ of $G'$,  $C_i$ is a contracted graph of $C_{i-1}$, and thus for each $i$, $C_i$ is a directed minor of $G'$ such that each vertex has one outgoing edge, and $C_i$ only contains a cycle.
  Note that $E_j$ forms a cycle of $C_j$ by induction. The algorithm outputs the directed cycle for each connected component found in Step 1.

  Now we show that Algorithm~\ref{algo:cycle} can be implemented in $\tilde{O}(1)$ Minor-Aggregation rounds.
  Finding all the connected components of $G'$ and identifying the connected components with a self-loop can be done in a single round by contracting all edges and doing a $\min$-aggregation to find the minimum ID in the component. 
  By the Chernoff bound, for any $C_i$,
  each directed path of length $O(\log n)$ contains a sampled vertex with probability $1 - O(1 / n^{10})$.
  By the union bound, with probability at least $1 - O(1/ n)$, for each $C_i$, each walk on $C_i$ either finds a cycle or hits a sampled vertex after $O(\log n)$ steps of the walk.
  Note that if two walks visit the same vertex, then their following walks are the same. 
  Hence, if a vertex is visited by different walks, we only keep the walk initiated by the vertex with the smallest ID.
  Thus, one can determine if there is a walk from a sampled vertex finding a cycle before hitting another sampled vertex in $O(\log n)$ Minor-Aggregation rounds, and return such a sampled vertex if exists. 

  By Corollary~\ref{corollary:computation-on-minors}, the vertex sampling and the minor construction can be implemented in $O(1)$ rounds, and thus Step 2(a) can be implemented in $O(1)$ rounds.  
  Step 2(b) can be implemented in $\tilde{O}(1)$ rounds using the observation that every iteration of Step 2(b) reduces the number of vertices in the connected component by a constant factor with constant probability. 
  Step 2(c) also can be implemented in $O(\log n)$ rounds by walking for $O(\log n)$ steps along the outgoing edges from a given vertex. 

  Note that for any $j \geq 1$, a vertex of $C_j$ belongs to $V_j$ if and only if the corresponding induced subgraph of $C_{j-1}$ contains two vertices of $V_{j-1}$. 
  For each induced subgraph of $C_{j-1}$ that corresponds to a vertex of $C_j$, it takes $O(1)$ rounds to identify the vertices in $V_{j-1}$, and another $O(1)$ rounds to identify the path connecting the two vertices in $V_{j-1}$ by Lemma~\ref{theorem:subtree-sum}. 
  Hence Step 2(d) can be simulated in $O(\log n)$ rounds.
\end{proof}

Now we are ready to show that Algorithm~\ref{algo:ESSSP} can be implemented in an efficient distributed manner. Roughly speaking, besides obtaining the solution to the transshipment problem, each step of Algorithm~\ref{algo:ESSSP} can be implemented in polylogarithmic number of rounds. 
Together with the round complexity of approximating transshipment (Theorem~\ref{theorem:distributed-transshipment}) and the fact that Algorithm~\ref{algo:ESSSP} has a recursion depth $O(\log n)$ (Lemma~\ref{lem:ESSSP}), we obtain the following lemma.

\begin{lemma}[Round complexity of ESSSP]\label{lemma:distribted-esssp-complexity}
  Suppose $G = (V, E, w)$ is a weighted graph, $s \in V$ is a distinguished node, $d \in \R^V$ is a demand vector satisfying $d_s \ge 0$ and $d_v \le 0$, and $\eps > 0$ is a parameter. There exists a distributed implementation of \Cref{algo:ESSSP} which outputs an $(1 + \eps)$-approximate ESSSP in $\eps^{-2} \cdot \exp(O(\log n \cdot \log \log n)^{3/4})$ rounds of Minor-Aggregations.
\end{lemma}
\begin{proof}
  Let $G_0 = (V_0, E_0)$ be the initial input graph of Algorithm~\ref{algo:ESSSP}, and $G_i = (V_i, E_i)$ be the input graph of $i$-th recursion for $i > 0$. 
  Note that $G_0$ is the same as $G$, and each $G_i$ is a minor of $G_0$.
  Hence, all the $G_i$ are minors of $G$.
  Throughout our implementation, for each $G_i$, we maintain a rooted forest $F_i$ on $G$
  such that 
  for each vertex $u$ of $G_i$, 
  there is a tree of $F_i$ corresponding to a spanning tree of $u$'s corresponding induced subgraph in $G$. 
  Since each vertex of $G_0$ is a vertex of $G$, $F_0$ only contains isolated vertices. 

  We analyze the round complexity to implement Algorithm~\ref{algo:ESSSP} (ESSSP) with respect to input graph $G_i$. 
  By Theorem~\ref{theorem:distributed-evaluation-routing}, Corollary~\ref{corollary:computation-on-minors} and Lemma~\ref{lem:mainTransshipment}, Line 1 can be implemented in $\eps^{-2} \exp(O(\log n \cdot \log \log n)^{3/4})$ rounds. 

  Now we show that Line 2 can be implemented in $O(1)$ rounds for each $G_i$. 
  First, for each vertex $u'$ in $G$,
  letting $u$ denote the vertex of $G_i$ whose corresponding induced subgraph in $G$ consisting $u'$, we compute 
  \[f_i(u') \defeq \sum_{v \in V_i: f(u, v) > 0 \text{ and  $(u, v)$ corresponds to an edge  of $G$ incident to $u'$}} f(u, v)\]
  in $O(1)$ rounds. 
  Second, we compute \[g_i(u')\defeq \sum_{v' \in V: v' \text{ is a descendant of $u'$ in $F_i$}} f_i(v')\] in $O(1)$ rounds by Lemma~\ref{theorem:subtree-sum}. 
  Third, each vertex $u'$ of $G$ samples a child $v'$ in $F_i$ with probability proportional to $g_i(v') / g_i(u')$, and no edge with probability $f_i(u') / g_i(u')$ in $O(1)$ rounds. 
  Forth, for each tree of $F_i$,
  we identify the vertex where the walk from the root  along edges sampled in third step stops. 
  For each tree $T$ in $F_i$, each vertex $u'$ in $T$ is sampled with probability $f_i(u') / \sum_{v' \in T} f_i(v')$.
  And this step can be simulated in $O(1)$ rounds by Lemma~\ref{theorem:subtree-sum}. 
  At the end, in $O(1)$ rounds, each vertex $u'$ of $G$ identified in the last step samples an edge $(u, v)$ with probability proportional to $f(u, v)$ among all the edges with positive flow value and corresponding to edges incident to $u'$. Since each vertex $u$ of $G_i$, edge $(u, v)$ with positive flow value is sampled with probability proportional to $f(u, v)$. Hence, Line 2 of Algorithm~\ref{algo:ESSSP} can be implemented in $O(1)$ rounds. 

  Let $H_i$ be the graph constructed in Line 2 of $i$-th recursion for $i > 0$. 
  The connected components of $H_i$ can be identified in a single round by contracting all edges and doing a $\min$-aggregation,
  and the directed cycles of $H_i$ can be found in $\tilde{O}(1)$ rounds by Corollary~\ref{corollary:computation-on-minors} and Lemma~\ref{lem:cycle}. 
  Line 3 to 12 be implemented in $O(1)$ rounds by Corollary~\ref{corollary:computation-on-minors} and Lemma~\ref{theorem:subtree-sum}.
  The construction of $G_{i+1}$ can be done in $O(1)$ rounds by Corollary~\ref{corollary:computation-on-minors}. The edges of $F_{i+1}$ can obtained by removing an arbitrary edge for each cycle in $H_i$, and taking the union with $F_{i}$, which can be simulated in $O(1)$ rounds. 
  By Corollary~\ref{corollary:computation-on-minors} and Lemma~\ref{lem:sssp_basic_operations}, the edge directions of $F_{i+1}$ can be obtained in $\tilde{O}(1)$ rounds.
  Line 15 to 19 can also be implemented in $O(1)$ rounds.

  Hence, besides the recursion, Algorithm~\ref{algo:ESSSP} on each input graph $G_i$ takes $\eps^{-2} \exp(O(\log n \cdot \log \log n)^{3/4})$ rounds. 
  Since the recursion depth of Algorithm~\ref{algo:ESSSP} is $O(\log n)$ by Lemma~\ref{lem:ESSSP}, Algorithm~\ref{algo:ESSSP} can be implemented in $\eps^{-2} \exp(O(\log n \cdot \log \log n)^{3/4})$ Minor-Aggregation rounds.  
\end{proof}

Now we give the distributed implementation of Algorithm~\ref{algo:TStoSSSP} (SSSP) and prove the main result of this section.

\thmSSSP*

\begin{proof}
  By Lemma~\ref{lem:TStoSSSP}, Algorithm~\ref{algo:TStoSSSP} outputs an SSSP with appropriate guarantees. In this proof, we show that Algorithm~\ref{algo:TStoSSSP} can be simulated in $\eps^{-2} \cdot \exp(O(\log n \cdot \log \log n)^{3/4})$ Minor-Aggregation rounds.

  Note that the while loop on lines 2--14 executes $O(\log n)$ w.h.p. (\Cref{lem:TStoSSSP}), hence it is sufficient to bound the complexity of a single iteration of the loop. Line 4 can be implemented in $\eps^{-2} \cdot \exp(O(\log n \cdot \log \log n)^{3/4})$ rounds via  \Cref{theorem:distributed-transshipment}. Line 5 has the same round complexity via \Cref{lemma:distribted-esssp-complexity}. Lines 6--8 are implemented in $\tilde{O}(1)$ rounds via Corollary~\ref{corollary:computation-on-minors} and \Cref{lem:sssp_basic_operations}. Other operations are trivially implementable in the Minor-Aggregation model.
\end{proof}

\bibliographystyle{alpha}
\bibliography{Refs}

\newpage
\appendix

\section{Deferred proofs}\label{sec:deferred-proofs}

\ThmAggregationCongest*
\begin{proof}
  The computation of an edge $e = \{u, v\}$ in the Minor-Aggregation model is simulated by both of its endpoints (they will always agree on the state of the edge). Naturally, the computation of a node in Minor-Aggregation model is simulated by the node itself in CONGEST. It is sufficient to show that we can simulate a single round of the Minor-Aggregation model in $\tilde{O}(\tau_{PA})$ rounds of the CONGEST model.
  
  \textbf{Leader election in each supernode.} We directly follow the argument laid out in \cite{GH16}. Initially, each node starts as its own \emph{cluster}. In subsequent iterations, we grow the clusters until the set of clusters matches the set of supernodes (i.e., there are no outgoing contracted edges from each cluster). Every node in each cluster $C \subseteq V$ maintains the minimum ID of the node in $C$, a node which we call the ``leader'' of $C$.

  We merge clusters together over contracted edges (i.e., $e \in E$ with $c_e = \top$). We do this by restricting the merges to be star shaped in the following way. The leader of each cluster $C$ throws a fair random coin and labels itself ``heads'' or ``tails''. This information is then propagated (via solving a PA task) to all nodes in $C$ in $\tilde{O}(\tau_{PA})$ rounds---following \Cref{def:pa-task}, the leader sets its private input to be $1$ or $0$ (corresponding to heads/tails), the aggregate operation is $\max$, each node uses the leader's ID as their part ID.

  Each tails cluster $C$ chooses an arbitrary edge connecting it to a heads cluster: (1) each node sends its leader ID to all neighbors, (2) each node sets its private input to be an arbitrary incident edge satisfying the condition (or $\bot$ if none), (3) solving a PA task inside each cluster we choose an arbitrary (e.g., one with minimal ID) such edge and inform each node in $C$ about the leader ID of the chosen heads cluster. Each node inside a tails cluster which has found a neighboring tails cluster takes over the part ID of the heads cluster.

  Consider a cluster $C$ that has an incident contracted edge $e$ (whose endpoints are in clusters $C$ and $D$). With probability at least $1/4$, if $C$ is tails and $D$ is heads, the cluster $C$ will dissapear (be merged into $D$). Therefore, using a standard argument, after $O(\log n)$ iterations, the clusters will exactly match the supernodes with high probability.
    
  \textbf{Contraction + Consensus step.} First, we elect a leader in each supernode. In other words, each node $v$ in a supernode $s \in V'$ agrees on some unique ID of that supernode. Then solving the distributed PA task with the leader's ID as the part ID, we can directly inform each node $v \in s$ of the aggregate $\bigoplus_{v \in s} x_v$ in $\tau_{PA}$ CONGEST rounds (\Cref{def:pa-task}).

  \textbf{Aggregation step.} Suppose that, after performing the consensus step, each node $v \in s$ is informed about $y_s$ and $v$'s supernode ID which we denote as $\mathrm{ID}(s)$. Each node $v \in s$ sends $\mathrm{ID}(s)$ and $y_s$ to all $v$'s neighbors. Consider an edge $\{a, b\} = e \in E$ and suppose they are in supernodes $a \in s_a, b \in s_b$. Both endpoints can (equally) simulate the computation on edge $e$: assuming $s_a \neq s_b$, $e$ corresponds to an edge $e \in E'$, in which case $a$ and $b$ compute $z_{e, s_a}$ and $z_{e, s_b}$. Now, each node $v \in s$ can compute the aggregate $\bigotimes_e z_{e, s}$ over all edges $e$ incident to $v$. Solving the distributed PA task, for each supernode $s$ we can compute the final aggregate $\bigoplus_{e \in \mathrm{incidentEdges}(s)} z_{e, s}$ and inform all $v \in s$ of this value.
\end{proof}

\lemmaDistributedLddSampling*
\begin{proof}
  The algorithm is presented in \cite{HL18}, in Section 2, Algorithm LDDSubroutine and Algorithm ExpectedSPForest. In order to avoid recreating large parts of \cite{HL18}, we will assume and reuse the notation of that paper in this proof.

  First, we observe that LDDSubroutine can be implemented in the Minor-Aggregation model. LDDSubroutine consists of (1) contracting 0-weighted nodes, and (2) performing multiple steps of ball-growing, where in each step each supernode becomes active if an neighboring supernode was active in the previous step. Contractions can be clearly performed via \Cref{corollary:computation-on-minors}. Ball-growing can also be simulated in the Minor-Aggregation model: each edge between supernode where one side is active and the other one is not will inform the inactive side to become active. Second, we also observe that ExpectedSPForest can be implemented using Minor-Aggregations in a trivial way since it simply calls LDDSubroutine multiple times.

  By reinterpreting Lemma 13, the result guarantees an LDD of quality $\frac{1}{\beta} n^{O(\log \log n) / \log(1/\beta)}$, and it runs in $\tilde{O}(\frac{1}{\beta})$ Minor-Aggregation rounds. Choosing $\beta = \exp(- O(\sqrt{\log n \cdot \log \log n}))$, we obtain the stated result.

  Specifically, we run the Algorithm ExpectedSPForest while $\rho^{(t)} = (\frac{O(\log n)}{\beta})^t$ is bounded by our desired radius $\rho_{\text{target}}$. In each step we find an LDD of radius $\frac{O(\log n)}{\beta}$ and quality $O(\log n)$; contract the components, increase non-contracted edges by $\frac{O(\log n)}{\beta}$ (so that the distance nodes in different components never decreases), and repeat the process on the contracted graph (one can operate on the contracted graph \Cref{corollary:computation-on-minors}). The returned decomposition is the one obtained in the final step of the process. This makes the radius property trivial since the final radius is at most $\rho^{(t)} \le \rho$, where $t \le \sqrt{\log n}$ is the number of iterations performed. The quality property is proven as follows. Let $G_i$ be the contracted graph in step $i$. Consider some nodes $x, y$. They are separated in the final components if they are separated at every step along the way. In the penultimate step $t-1$, the expected distance between them is $\E[d_{G_{t-1}}(x, y)] \le O(\log n)^{t-1} d_G(x, y)$ (Lemma 12). Therefore, the probability of them being cut in the final step is at most $\E[ \frac{d_{G_{t-1}}(x, y)}{\rho^{(t)}} \cdot O(\log n) ] \le \frac{\E[d_{G_{t-1}}(x, y)]}{\rho^{(t)}} \cdot O(\log n) \le \frac{d_G(x, y) \cdot O(\log n)^{t-1}}{\rho^{(t)}} \cdot O(\log n) = \frac{d_G(x, y)}{\rho^{(t)}} \cdot O(\log n)^t \le \frac{d_G(x, y)}{\rho} \cdot \frac{O(\log n)}{\beta} \cdot O(\log n)^t \le \frac{d_G(x, y)}{\rho} \cdot \exp(O(\sqrt{\log n \cdot \log \log n}))$. Therefore, by definition of quality, it is at most $\exp(O(\sqrt{\log n \cdot \log \log n}))$.%
\end{proof}

\end{document}